\documentclass{llncs}

\usepackage{paralist} 
\usepackage{stmaryrd} 
\usepackage{amsmath}
\usepackage{amssymb}
\usepackage{epsfig} 
\usepackage{epstopdf}
\usepackage{float}
\usepackage{graphicx}
\usepackage{caption, subcaption}
\usepackage{wrapfig}
\usepackage[hyphens]{url}

\newcommand{\paper}[1]{}
\usepackage{tikz}
\usetikzlibrary{shapes,arrows,fit,calc,positioning,automata}

\newcommand{\nonl}{\renewcommand{\nl}{\let\nl\oldnl}}

\makeatletter
\newif\if@restonecol
\makeatother

\usepackage[ruled,vlined,linesnumbered,lined]{algorithm2e}

\usepackage{xcolor,colortbl}
\definecolor{Gray}{gray}{0.95}
\definecolor{LightCyan}{rgb}{0.88,1,1}

 \usepackage{relsize}

\newcommand{\fct}{\ensuremath{\mathsf{fact}}}

\newcommand{\formula}[1]{\ensuremath{{\Phi(#1)}}}
\newcommand{\err}{\ensuremath{\mathsf{Err}}}

\newcommand{\FF}{\ensuremath{\overrightarrow{F}}}
\newcommand{\eat}[1]{}
\newcommand{\op}{\ensuremath{\mathsf{op}}}

\newcommand{\subst}[2]{#1\mapsto #2} 

\newcommand{\mcomment}[1]{}

\newcommand{\ite}{\ensuremath{\mathsf{ite}}}

\newcommand{\cba}[1]{{\Delta_{#1}}}
\newcommand{\cbb}[1]{{\Gamma_{#1}}}

\newcommand{\cbva}[1]{\overrightarrow{\Delta}(#1)}
\newcommand{\cbvb}[1]{\overrightarrow{\Gamma}(#1)}
\newcommand{\cbvasimp}{\overrightarrow{\Delta}}
\newcommand{\cbvbsimp}{\overrightarrow{\Gamma}}

\newcommand{\cbar}[1]{\delta_{#1}}
\newcommand{\cbbr}[1]{\gamma_{#1}}

\newcommand{\true}{\textsf{true}}
\newcommand{\false}{\textsf{false}}

\newcommand{\bparcegar}{\mathsf{\mathbf{ParSyn}}}

\newcommand{\parcegar}{\mathsf{ParSyn}}

\newcommand{\cegarsk}{\mathsf{CSk}}
\newcommand{\bcegarsk}{\mathsf{\mathbf{CSk}}}
\newcommand{\bloqqer}{\textsf{Bloqqer}}
\newcommand{\rsynth}{\textsf{RSynth}}
\newcommand{\qrat}{\textsf{qrat-trim}}

\usetikzlibrary{decorations.pathreplacing} 
\usetikzlibrary{patterns}

\usepackage[ruled,vlined,linesnumbered,lined]{algorithm2e}

\tikzset{
block/.style={
  draw, 
  rectangle, 
  minimum height=0cm, 
  minimum width=2cm, align=center,
  }, 
line/.style={->,>=latex'}
}
\tikzstyle{galivenode}=[circle,fill=green!50!black,thick,inner sep=1pt,minimum size=4mm]
\tikzstyle{ralivenode}=[circle,fill=red!70!black,thick,inner sep=1pt,minimum size=4mm]

\tikzstyle{alivenode}=[circle,fill=black!80,thick,inner sep=1pt,minimum size=4mm]
\tikzstyle{deadnode}=[circle,fill=black!10,thick,inner sep=0pt,minimum size=4mm]
\tikzstyle{rdnode}=[circle,draw,fill=red!10,thick,inner sep=2pt,minimum size=4mm]
\tikzstyle{grnode}=[circle,draw,fill=green!10,thick,inner sep=2pt,minimum size=4mm]

\tikzstyle{ynode}=[rectangle,draw=black,fill=yellow!10,thick,inner sep=3pt,minimum size=4mm]
\tikzstyle{gnode}=[rectangle,fill=green!10,thick,inner sep=1pt,minimum size=4mm]
\tikzstyle{bnode}=[rectangle,fill=red!10,thick,inner sep=1pt,minimum size=4mm]

\begin{document}

\setlength{\pdfpageheight}{\paperheight}
\setlength{\pdfpagewidth}{\paperwidth}

\title{Towards Parallel Boolean Functional Synthesis}

\author{S. Akshay\inst{1} \and Supratik Chakraborty\inst{1} \and Ajith K. John\inst{2} \and Shetal Shah\inst{1}}

\institute{IIT Bombay, India
  \and
HBNI, BARC, India}

\maketitle

\begin{abstract}
Given a relational specification $\varphi(X, Y)$, where $X$ and $Y$
are sequences of input and output variables, we wish to synthesize
each output as a function of the inputs such that the specification
holds. This is called the Boolean functional synthesis problem and has
applications in several areas. In this paper, we present the first
parallel approach for solving this problem, using compositional and
CEGAR-style reasoning as key building blocks. We show by means of
extensive experiments that our approach outperforms existing tools on
a large class of benchmarks.
\end{abstract}


\section{Introduction}
\label{sec:introduction}
Given a relational specification of input-output behaviour,
synthesizing outputs as functions of inputs is a key step in
several applications, viz. program repair~\cite{JGB05}, program
synthesis~\cite{Gulwani}, adaptive control~\cite{RW87} etc.  The
synthesis problem is, in general, uncomputable.  However, there are
practically useful restrictions that render the problem solvable, e.g., if all inputs and outputs are Boolean, the problem is
computable in principle.  Nevertheless, functional synthesis may still
require formidable computational effort, especially if there are a
large number of variables and the overall specification is
complex.  This motivates us to investigate techniques for Boolean
functional synthesis that work well in practice.

Formally, let $X$ be a sequence of $m$ input Boolean variables, and
$Y$ be a sequence of $n$ output Boolean variables.  A relational
specification is a Boolean formula $\varphi(X, Y)$ that expresses a
desired input-output relation.  The goal in Boolean functional
synthesis is to synthesize a function
$F: \{0,1\}^m\rightarrow \{0,1\}^n$ that satisfies the specification.
Thus, for every value of $X$, if there exists some value of $Y$ such
that $\varphi(X,Y)=1$, we must also have $\varphi(X,F(X)) = 1$.  For
values of $X$ that do not admit any value of $Y$ such that
$\varphi(X,Y) = 1$, the value of $F(X)$ is inconsequential.  Such a
function $F$ is also refered to as a \emph{Skolem function} for $Y$ in
$\varphi(X,Y)$~\cite{bierre,fmcad2015:skolem}.

An interesting example of Boolean functional synthesis is the problem
of integer factorization.  Suppose $Y_1$ and $Y_2$ are $n$-bit
unsigned integers, $X$ is a $2n$-bit unsigned integer and
$\times_{[n]}$ denotes $n$-bit unsigned multiplication.  The
relational specification $\varphi_{\fct}(X, Y_1, Y_2) \equiv
((X=Y_1\times_{[n]} Y_2) \wedge (Y_1\neq 1) \wedge (Y_2\neq 1))$
specifies that $Y_1$ and $Y_2$ are non-trivial factors of $X$.  This
specification can be easily encoded as a Boolean relation. 
 The corresponding synthesis problem requires
us to synthesize the factors $Y_1$ and $Y_2$ as functions of $X$,
whenever $X$ is non-prime. Note that this problem is known to be hard,
and the strength of several cryptographic systems rely on this
hardness.

Existing approaches to Boolean functional synthesis vary widely in
their emphasis, ranging from purely theoretical treatments
(viz.~\cite{boole1847,lowenheim1910,deschamps1972,boudet1989,martin1989,baader1999})
to those motivated by practical tool development
(viz.~\cite{macii1998,bierre,Jian,Trivedi,fmcad2015:skolem,rsynth,KS00,BCK09,Gulwani,SRBE05,KMPS10}). A
common aspect of these approaches is their focus on sequential
algorithms for synthesis.  In this paper, we present, to the best of
our knowledge, the first parallel algorithm for Boolean functional
synthesis. A key ingredient of our approach is a technique for solving
the synthesis problem for a specification $\varphi$ by composing
solutions of synthesis problems corresponding to sub-formulas in
$\varphi$.  Since Boolean functions are often represented using
DAG-like structures (such as circuits, AIGs~\cite{aiger},
ROBDDs~\cite{akers1978,bryant1986}), we assume w.l.o.g. that $\varphi$
is given as a DAG.  The DAG structure provides a natural decomposition
of the original problem into sub-problems with a partial order of
dependencies between them.  We exploit this to design a parallel
synthesis algorithm that has been implemented on a message passing
cluster.  Our initial experiments show that our algorithm
significantly outperforms state-of-the-art techniques on several
benchmarks.

\paragraph{Related work:}
The earliest solutions to Boolean functional synthesis date back to
Boole~\cite{boole1847} and Lowenheim~\cite{lowenheim1910}, who
considered the problem in the context of Boolean unification.
Subsequently, there have been several investigations into theoretical
aspects of this problem (see
e.g.,~\cite{deschamps1972,boudet1989,martin1989,baader1999}).  More
recently, there have been attempts to design practically efficient
synthesis algorithms that scale to much larger problem
sizes. In~\cite{bierre}, a technique to synthesize $Y$ from a proof of
validity of $\forall X \exists Y \varphi(X,Y)$ was proposed.  While
this works well in several cases, not all specifications admit the
validity of $\forall X \exists Y \varphi(X,Y)$.  For example, $\forall
X \exists Y \varphi_{\fct}(X,Y)$ is not valid in the factorization
example.  In~\cite{Jian,Trivedi}, a synthesis approach based on
functional composition was proposed.  Unfortunately, this does not
scale beyond small problem instances~\cite{fmcad2015:skolem,rsynth}.
To address this drawback, a CEGAR based technique for synthesis
from \emph{factored} specifications was proposed
in~\cite{fmcad2015:skolem}.  While this scales well if each factor in
the specification depends on a small subset of variables, its
performance degrades significantly if we have a few ``large'' factors,
each involving many variables, or if there is significant sharing of
variables across factors.  In~\cite{macii1998}, Macii et al
implemented Boole's and Lowenheim's algorithms using ROBDDs and
compared their performance on small to medium-sized benchmarks.  Other
algorithms for synthesis based on ROBDDs have been investigated
in~\cite{KS00,BCK09}.  A recent work~\cite{rsynth} adapts the
functional composition approach to work with ROBDDs, and shows that
this scales well for a class of benchmarks with pre-determined
variable orders.  However, finding a good variable order for an
arbitrary relational specification is hard, and our experiments show
that without prior knowledge of benchmark classes and corresponding
good variable orders, the performance of~\cite{rsynth} can degrade
significantly. Techniques using \emph{templates}~\cite{Gulwani}
or \emph{sketches}~\cite{SRBE05} have been found to be effective for
synthesis when we have partial information about the set of candidate
solutions. A framework for functional synthesis, focused on unbounded
domains such as integer arithmetic, was proposed in~\cite{KMPS10}.
This relies heavily on tailor-made smart heuristics that exploit specific
form/structure of the relational specification.

\section{Preliminaries}
\label{sec:preliminaries}
Let $X = (x_1, \ldots x_m)$ be the sequence of input variables, and $Y
= (y_1, \ldots y_n)$ be the sequence of output variables in the
specification $\varphi(X,Y)$.  Abusing notation, we use $X$
(resp. $Y$) to denote the set of elements in the sequence $X$
(resp. $Y$), when there is no confusion.  We use $1$ and $0$ to denote
the Boolean constants {\true} and {\false}, respectively.
A \emph{literal} is either a variable or its complement. An assignment
of values to variables \emph{satisfies} a formula if it makes the formula {\true}.

\begin{wrapfigure}[12]{r}{0.6\textwidth}
\vspace*{-0.3in}
\scalebox{0.8}{
  \begin{tikzpicture}[sibling distance=5em,
  main node/.style = {shape=circle,    draw, align=center},
  tri node/.style={shape=regular polygon, regular polygon sides=3, draw, scale=0.8},    
  edge from parent/.style={draw},
  no edge below/.style={    
        every child/.append style={solid,       
         edge from parent/.style={dashed}}},
  level 2/.style={level distance=8mm}
  level 1/.style={level distance=8mm},
  level 3/.style={sibling distance=2em} 
]
    \node[main node] (1) {$\wedge$}
  child {
    node[main node] (2) {$\vee$}
    child {
      node[main node] (5) {$\wedge$}
      child{node (x1){$x_1$}}
      child{node (y1){$y_1$}}
    }
    child{
      node[main node] (6) {$\vee$}
      }
  }
  child { node[main node] (3) {$\wedge$}
    child {node[main node] (6) {$\vee$}
      child{node (x2){$\neg x_2$}}
      child{node (y2){$0$}}
    }
    child {node[main node] (7) {$\wedge$}
      child {node (x3){$x_3$}}
      child {node (y3){$\neg y_3$}}
      child{node (t) {$\ldots$}}
    }
  }
  child { node {$\ldots$}[no edge below]}
  child {
    node[main node] (4) {$\wedge$}
    child{node (s){$1$}}
    child {node (8) {$x_{m-1}$}}
    child {node[main node] (9) {$\wedge$}
      child{node (x4){$x_m$}}
      child{node (y4){$y_n$}}
    }
  }
  ;
\end{tikzpicture}
}
\caption{DAG representing $\varphi(X,Y)$}
\label{fig:decomp}
\end{wrapfigure}
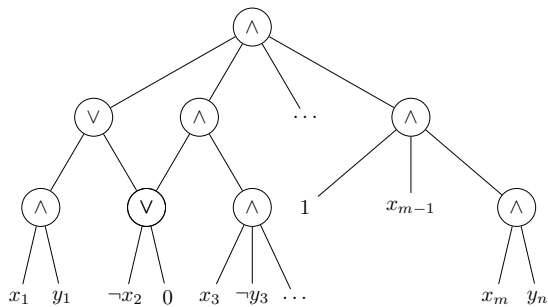
We assume that the specification $\varphi(X,Y)$ is represented as a
rooted DAG, with internal nodes labeled by Boolean operators and
leaves labeled by input/output literals and Boolean constants.  If
the operator labeling an internal node $N$ has arity $k$, we assume
that $N$ has $k$ ordered children.  Fig.~\ref{fig:decomp} shows an
example DAG, where the internal nodes are labeled by AND and
OR operators of different arities. Each node $N$ in such a DAG
represents a Boolean formula $\formula{N}$, which is inductively
defined as follows. If $N$ is a leaf, $\formula{N}$ is the label of
$N$.  If $N$ is an internal node labeled by $\op$ with arity $k$, and
if the ordered children of $N$ are $c_1, \ldots c_k$, then
$\formula{N}$ is $\op(\formula{c_1}, \ldots \formula{c_k})$.  A DAG
with root $R$ is said to represent the formula $\formula{R}$.  Note
that popular DAG representations of Boolean formulas, such as AIGs,
ROBDDs and Boolean circuits, are special cases of this representation.

A $k$-ary Boolean function $f$ is a mapping from $\{0,1\}^k$ to
$\{0,1\}$, and can be viewed as the semantics of a Boolean formula
with $k$ variables.  We use the terms ``Boolean function'' and
``Boolean formula'' interchangeably, using
formulas mostly to refer to specifications.
Given a Boolean formula $\varphi$ and a Boolean function $f$, we use
$\varphi[\subst{y}{f}]$ to denote the formula obtained by substituting
every occurrence of the variable $y$ in $\varphi$ with $f$.  The set
of variables appearing in $\varphi$ is called the \emph{support} of
$\varphi$.  If $f$ and $g$ are Boolean functions, we say that $f$
\emph{abstracts} $g$ and $g$ \emph{refines} $f$, if
$g \rightarrow f$, where $\rightarrow$ denotes logical implication.
 
Given the specification $\varphi(X, Y)$, our goal is to synthesize the
outputs $y_1, \ldots y_n$ as functions of $X$.  Unlike some earlier
work~\cite{bierre,skizzo,jiang2}, we \emph{do not assume the validity
of $\forall X \exists Y~ \varphi(X,Y)$}.  Thus, we allow the
possibility that for some values of $X$, there may be no value of $Y$
that satisfies $\varphi(X, Y)$.  This allows us to accommodate some
important classes of synthesis problems, viz. integer factorization.
If $y_1=f_1(X), \ldots y_n=f_n(X)$ is a solution to the synthesis
problem, we say that $(f_1(X), \ldots f_n(X))$ \emph{realizes} $Y$ in
$\varphi(X, Y)$.  For notational clarity, we simply use $(f_1, \ldots
f_n)$ instead of $(f_1(X), \ldots f_n(X))$ when $X$ is clear from the
context.

In general, an instance of the synthesis problem may not have a unique
solution.  The following proposition, stated in various forms in the
literature, characterizes the space of all solutions, when we have one
output variable $y$.
\begin{proposition}
\label{prop:soln-space}
A function $f(X)$ realizes $y$ in $\varphi(X, y)$ iff the following holds:\\
$\varphi[\subst{y}{1}] \wedge \neg \varphi[\subst{y}{0}] 
~\rightarrow~ f(X)$ and $f(X) ~\rightarrow~
\varphi[\subst{y}{1}] \vee \neg \varphi[\subst{y}{0}]$.
\end{proposition}
As a corollary, both $\varphi[\subst{y}{1}]$ and
$\neg \varphi[\subst{y}{0}]$ realize $y$ in $\varphi(X,y)$.
Proposition~\ref{prop:soln-space} can be easily extended when we have
multiple output variables in $Y$.  Let $\sqsubseteq$ be a total
ordering of the variables in $Y$, and assume without loss of generality that $y_1 \sqsubseteq
y_2 \sqsubseteq \cdots y_n$.  
Let $\FF$ denote the vector of Boolean functions $(f_1(X), \ldots
f_n(X))$.  For $i \in \{1, \ldots
n\}$, define $\varphi^{(i)}$ to be $\exists y_1 \ldots
\exists y_{i-1}\, \varphi$, and $\varphi^{(i)}_{\FF}$ to be 
$(\cdots (\varphi^{(i)}[\subst{y_{i+1}}{f_{i+1}}]) \cdots
  )[\subst{y_n}{f_n}]$, with the obvious modifications for $i = 1$ (no
  existential quantification) and $i = n$ (no substitution).  The
  following proposition, once again implicit in the literature,
  characterizes the space of all solutions ${\FF}$ that realize $Y$ in
  $\varphi(X,Y)$.
\begin{proposition}
\label{prop:soln-space-multi}
The function vector $\FF = (f_1(X), \ldots f_n(X))$ realizes $Y = (y_1, \ldots
  y_n)$ in $\varphi(X, Y)$ iff the following holds for every
  $i \in \{1, \ldots n\}$:\\
$\varphi^{(i)}_{\FF}[\subst{y_i}{1}] \wedge \neg \varphi^{(i)}_{\FF}[\subst{y_i}{0}] 
\rightarrow f_i(X)$, and
$f_i(X) \rightarrow
\varphi^{(i)}_{\FF}[\subst{y_i}{1}] \vee \neg \varphi^{(i)}_{\FF}[\subst{y_i}{0}].
$
\end{proposition}

Propositions~\ref{prop:soln-space} and \ref{prop:soln-space-multi} are
effectively used in~\cite{Jian,Trivedi,fmcad2015:skolem,rsynth} to
sequentially synthesize $y_1, \ldots y_n$ as functions of $X$.
Specifically, output $y_1$ is first synthesized as a function $g_1(X,
y_2, \ldots y_n)$.  This is done by treating $y_1$ as the sole output
and $X \cup \{y_2,\ldots y_n\}$ as the inputs in $\varphi(X, Y)$.  By
substituting $g_1$ for $y_1$ in $\varphi$, we obtain
$\varphi^{(2)} \equiv \exists y_1 \varphi(X, Y)$.  Output $y_2$ can
then be synthesized as a function $g_2(X, y_3, \ldots y_n)$ by
treating $y_2$ as the sole output and $X \cup \{y_3, \ldots y_n\}$ as
the inputs in $\varphi^{(2)}$.  Substituting $g_2$ for $y_2$ in
$\varphi^{(2)}$ gives $\varphi^{(3)}
\equiv \exists y_1 \exists y_2\, \varphi(X, Y)$.  This process is then
repeated until we obtain $y_n$ as a function $g_n(X)$.  The desired
functions $f_1(X), \ldots f_n(X)$ realizing $y_1, \ldots y_n$ can now
be obtained by letting $f_n(X)$ be $g_n(X)$, and $f_i(X)$ be $(\cdots
(g_i[\subst{y_{i+1}}{f_{i+1}(X)}]) \cdots)[\subst{y_n}{f_n(X)}]$, for
all $i$ from $n-1$ down to $1$.  Thus, given $\varphi(X,Y)$, it
suffices to obtain $(g_1, \ldots g_n)$, where $g_i$ has support
$X \cup \{y_{i+1}, \ldots y_n\}$, in order to solve the synthesis
problem.  We therefore say that $(g_1, \ldots g_n)$
\emph{effectively realizes} $Y$ in $\varphi(X, Y)$, and focus on obtaining
$(g_1, \ldots g_n)$.

Proposition~\ref{prop:soln-space} implies that for every
$i \in \{1, \ldots n\}$, the function
$g_i \equiv \varphi^{(i)}[\subst{y_i}{1}]$ realizes $y_{i}$ in
$\varphi^{(i)}$.  With this choice for $g_i$, it is easy to see that
$\exists y_i\, \varphi^{(i)}$ (or $\varphi^{(i+1)}$) can be obtained
as $\varphi^{(i)}[\subst{y_i}{g_i}]
= \varphi^{(i)}[\subst{y_i}{\varphi^{(i)}[\subst{y_i}{1}]}]$.  While
synthesis using quantifier elimination by
such \emph{self-substitution}~\cite{rsynth} has been shown to scale
for certain classes of specifications with pre-determined optimized
variable orders, our experience shows that this incurs significant
overheads for general specifications with unknown ``good'' variable
orders.  An alternative technique for \emph{factored} specification
was proposed by John et al~\cite{fmcad2015:skolem}, in which initial
abstractions of $g_1, \ldots g_n$ are first computed quickly, and then
a CEGAR-style~\cite{CEGAR-JACM} loop is used to refine these
abstractions to correct Skolem functions.  We use John et al's refinement
technique as a black-box module in our work; more on this is discussed
in Section~\ref{sec:simple}.
\begin{definition}
\label{def:cb0-cb1}
Given a specification $\varphi(X, Y)$, we define $\cba{y_i}(\varphi)$
to be the formula $\left(\neg \exists y_1 \ldots
y_{i-1}\, \varphi\right)[y_i \mapsto 0]$, and $\cbb{y_i}(\varphi)$ to
be the formula $\left(\neg \exists y_1
\ldots y_{i-1}\, \varphi\right)[y_i \mapsto 1]$,
for all $i \in \{1, \ldots n\}$\footnote{In~\cite{fmcad2015:skolem},
equivalent formulas were called $Cb0_{y_i}(\varphi)$ and
$Cb1_{y_i}(\varphi)$.}.  We also define $\cbva{\varphi}$ and $\cbvb{\varphi}$
to be the vectors  $(\cba{y_1}(\varphi), \ldots \cba{y_n}(\varphi))$
and $(\cbb{y_1}(\varphi), \ldots \cbb{y_n}(\varphi))$ respectively.
\end{definition}
If $N$ is a node in the DAG representation of the specification,
we abuse notation and use $\cba{y_i}(N)$ to denote
$\cba{y_i}(\formula{N})$, and similarly for $\cbb{y_i}(N)$, $\cbva{N}$
and $\cbvb{N}$.  Furthermore, if both $Y$ and $N$ are clear from the
context, we use $\cba{i}$, $\cbb{i}$, $\cbvasimp$ and $\cbvbsimp$
instead of $\cba{y_i}(N)$, $\cbb{y_i}(N)$, $\cbva{N}$ and $\cbvb{N}$,
respectively.  It is easy to see that the supports of both $\cbb{i}$
and $\cba{i}$ are (subsets of) $X \cup \{y_{i+1}, \ldots y_n\}$.
Furthermore, it follows from Definition~\ref{def:cb0-cb1} that whenever
$\cbb{i}$ (resp. $\cba{i}$) evaluates to $1$, if the output $y_i$ 
has the value $1$ (resp. $0$), then $\varphi$ must evaluate to $0$.
Conversely, if $\cbb{i}$ (resp. $\cba{i}$) evaluates to $0$, it
doesn't hurt (as far as satisfiability of $\varphi(X,Y)$ is concerned)
to assign the value $1$ (resp. $0$) to output $y_i$.  This suggests
that both $\neg \cbb{i}$ and $\cba{i}$ suffice to serve as the
function $g_i(X, y_{i+1}, \ldots y_n)$ when synthesizing functions for
multiple output variables.
The following proposition, adapted from~\cite{fmcad2015:skolem},
follows immediately, where we have abused notation and used
$\neg\cbvbsimp$ to denote $(\neg \cbb{1}, \ldots \neg \cbb{n})$.
\begin{proposition}
\label{prop:cb1-realizes} 
Given a specification $\varphi(X, Y)$, 
both $\cbvasimp$ and 
$\neg\cbvbsimp$ effectively realize $Y$ in
$\varphi(X,Y)$.
\end{proposition}
Proposition~\ref{prop:cb1-realizes} shows that it suffices to compute
$\cbvasimp$ (or $\cbvbsimp$) from $\varphi(X,Y)$ in order to solve the
synthesis problem.  In the remainder of the paper, we show how to
achieve this compositionally and in parallel by first computing
refinements of $\cba{i}$ (resp. $\cbb{i}$) for all $i \in \{1,\ldots
n\}$, and then using John et al's CEGAR-based
technique~\cite{fmcad2015:skolem} to abstract them to the desired
$\cba{i}$ (resp. $\cbb{i}$).  Throughout the paper, we use
$\cbar{i}$ and $\cbbr{i}$ to denote refinements of $\cba{i}$ and
$\cbb{i}$ respectively.

\section{Exploiting compositionality}
\label{sec:compose}
Given a specification $\varphi(X,Y)$, one way to synthesize $y_1,
\ldots y_n$  is to decompose $\varphi(X,Y)$ into
sub-specifications, solve the synthesis problems for the
sub-specifications in parallel, and compose the solutions to the
sub-problems to obtain the overall solution.  A DAG representation of
$\varphi(X,Y)$ provides a natural recursive decomposition of the
specification into sub-specifications.  Hence, the key technical
question relates to compositionality: how do we compose solutions to
synthesis problems for sub-specifications to obtain a solution to the
synthesis problem for the overall specification?  This problem is not
easy, and no state-of-the-art tool for Boolean functional synthesis
uses such compositional reasoning.

Our compositional solution to the synthesis problem is best explained
in three steps.  First, for a simple, yet representationally complete,
class of DAGs representing $\varphi(X,Y)$, we present a lemma that
allows us to do compositional synthesis at each node of such a DAG.
Next, we show how to use this lemma to design a parallel synthesis
algorithm. Finally, we extend our lemma, and hence the scope of our
algorithm, to significantly more general classes of DAGs.

\subsection{Compositional synthesis in AND-OR DAGs}
\label{sec:simple}
For simplicity of exposition, we first consider DAGs with internal
nodes labeled by only AND and OR operators (of arbitrary arity). 
Fig.~\ref{fig:decomp} shows an example of such a DAG.  Note that this
class of DAGs is representationally complete for Boolean
specifications, since every specification can be expressed in negation
normal form (NNF).  In the previous section, we saw that computing
$\cba{i}(\varphi)$ or $\cbb{i}(\varphi)$ for all $i$ in $\{1, \ldots
n\}$ suffices for purposes of synthesis.
The following lemma shows the relation between $\cba{i}$ and $\cbb{i}$
at an internal node $N$ in the DAG and the corresponding formulas at
the node's children, say $c_1, \ldots c_k$.
\begin{lemma}[Composition Lemma]
\label{prop:comp}
Let $\Phi(N) = \op(\Phi(c_1),\ldots, \Phi(c_k))$, where $\op=\vee$ or
$\op=\wedge$. Then, for each $1\leq i\leq n$:
\begin{align}
  \left(\bigwedge_{j=1}^k \cba{i}(c_j)\right) \leftrightarrow \cba{i}(N) \quad\mbox{and}\quad
  \left(\bigwedge_{j=1}^k \cbb{i}(c_j)\right) \leftrightarrow \cbb{i}(N) \mbox{ if } \op = \vee\label{eq:vee}\\
  \left(\bigvee_{j=1}^k \cba{i}(c_j)\right) \rightarrow \cba{i}(N) \quad\mbox{and}\quad
  \left(\bigvee_{j=1}^k \cbb{i}(c_j)\right) \rightarrow \cbb{i}(N) \mbox{ if } \op = \wedge\label{eq:wedge}
\end{align}
\end{lemma}
\begin{proof} 
The proof of this lemma follows from Definition~\ref{def:cb0-cb1}.
Consider the case of disjunction $\op=\vee$, i.e., Equation~(\ref{eq:vee}) for $\cba{}$ (the case for $\cbb{}$ is similar). Then 
  \begin{align*}
    \cba{i}(N) &= \neg \exists y_1\ldots y_{i-1} (\Phi(c_1)\vee \ldots \vee\Phi(c_k))[y_i\mapsto 0]\\
    &\longleftrightarrow \forall y_1\ldots y_{i-1}(\neg\Phi(c_1)\wedge \ldots \wedge \neg \Phi(c_k))[y_i\mapsto 0]\\
    &\longleftrightarrow (\forall y_1\ldots y_{i-1}\neg\Phi(c_1))[y_i\mapsto 0]\wedge \ldots \wedge (\forall y_1\ldots y_{i-1}\neg \Phi(c_k))[y_i\mapsto 0]\\
    &\longleftrightarrow \cba{i}(c_1)\wedge\ldots \wedge\cba{i}(c_k)
    \end{align*}

  On the other hand for conjunction $\op=\wedge$, i.e., Equation~(\ref{eq:wedge}) (and similarly for $\cbb{}$), we only have one direction:
  \begin{align*}
    \cba{i}(N) &= \neg \exists y_1\ldots y_{i-1} (\Phi(c_1)\wedge \ldots \wedge\Phi(c_k))[y_i\mapsto 0]\\
    &\longleftrightarrow \forall y_1\ldots y_{i-1}(\neg\Phi(c_1)\vee \ldots \vee \neg \Phi(c_k))[y_i\mapsto 0]\\
    &\longleftarrow(\forall y_1\ldots y_{i-1}\neg\Phi(c_1))[y_i\mapsto 0]\vee \ldots \vee (\forall y_1\ldots y_{i-1}\neg \Phi(c_k))[y_i\mapsto 0]\\
    &\longleftrightarrow \cba{i}(c_1)\vee\ldots \vee\cba{i}(c_k)
  \end{align*}
  This completes the proof.\qed
   \end{proof}

Thus, if $N$ is an OR-node, we obtain $\cba{i}(N)$ and $\cbb{i}(N)$
directly by conjoining $\cba{i}$ and $\cbb{i}$ at its
children. However, if $N$ is an AND-node, disjoining the $\cba{i}$ and
$\cbb{i}$ at its children only gives refinements of $\cba{i}(N)$ and
$\cbb{i}(N)$ (see Equation~(\ref{eq:wedge})). Let us call these
refinements $\cbar{i}(N)$ and $\cbbr{i}(N)$ respectively.  To obtain
$\cba{i}(N)$ and $\cbb{i}(N)$ exactly at AND-nodes, we must use the
CEGAR technique developed in~\cite{fmcad2015:skolem} to iteratively
abstract $\cbar{i}(N)$ and $\cbbr{i}(N)$ obtained above.  More on
this is discussed below.

A CEGAR step involves constructing, for each $i$ from $1$ to $n$, a
Boolean \emph{error formula} $\err_{\cbar{i}}$
(resp. $\err_{\cbbr{i}}$) such that the error formula is unsatisfiable
iff $\cbar{i}(N) \leftrightarrow \cba{i}(N)$ (resp. $\cbbr{i}(N)
\leftrightarrow \cbb{i}(N)$).  A SAT solver is then used to check the
satisfiability of the error formula.  If the formula is unsatisfiable,
we are done; otherwise the satisfying assignment can be used to
further abstract the respective refinement.  This check-and-abstract
step is then repeated in a loop until the error formulas become
unsatisfiable.  Following the approach outlined
in~\cite{fmcad2015:skolem}, it can be shown that if we use
$\err_{\cbar{i}} ~\equiv~ \neg \cbar{i}
~\wedge~ \bigwedge_{j=1}^{i}\left(y_j \leftrightarrow \cbar{j}\right)
~\wedge~ \neg \varphi$ and $\err_{\cbbr{i}} ~\equiv~ \neg \cbbr{i}
~\wedge~ \bigwedge_{j=1}^{i}\left(y_j \leftrightarrow \neg\cbbr{j}\right)
~\wedge~ \neg \varphi$, and perform CEGAR in order from $i = 1$ to
$i=n$, it suffices to gives us $\cba{i}$ and $\cbb{i}$.  For details
of the CEGAR implementation, the reader is referred
to~\cite{fmcad2015:skolem}.  The above discussion leads to a
straightforward algorithm {\textsc{Compute}} (shown as
Algorithm~\ref{alg:composeatop}) that computes $\cbva{N}$ and
$\cbvb{N}$ for a node $N$, using $\cbva{c_j}$ and $\cbvb{c_j}$ for its
children $c_j$.  Here, we have assumed access to a black-box function
{\textsc{Perform\_Cegar}} that implements the CEGAR step.

\begin{algorithm}[h]
\caption{\textsc{Compute(Node $N$)}}
 \label{alg:composeatop}
 \KwIn{A DAG Node $N$ labelled either AND or OR \\ 
 \textbf{Precondition:} Children of $N$, if any, have their  $\cbvasimp$ and $\cbvbsimp$ computed.}
\KwOut{ $\cbva{N}, \cbvb{N}$} 


\If {$N$ is a leaf \hfill \tcp{$\formula{N}$ is a literal/constant; use Definition~\ref{def:cb0-cb1}}} 
	    {
              for all $y_i\in Y$, $\cba{i}(N)= \neg \exists y_1\ldots y_{i-1} (\formula{N}) [y_i\mapsto 0]$\;
              for all $y_i\in Y$, $\cbb{i}(N)= \neg \exists y_1\ldots y_{i-1} (\formula{N}) [y_i\mapsto 1]$\;
	}
	\Else{\tcp{$N$ is an internal node; let its children be $c_1, \ldots c_k$} 
                  
		\If {$N$ is an OR-node}
		{
			\For {\textbf {each} $y_i \in Y$}
			{
				$\cba{i}(N) :=  \cba{i}(c_1) \wedge \cba{i}(c_2)\ldots \wedge \cba{i}(c_k)$\;
				$\cbb{i}(N) :=  \cbb{i}(c_1)  \wedge \cbb{i}(c_2)\ldots \wedge \cbb{i}(c_k)$\;
			}
                }
\If {$N$ is an AND-node}
		{
			\For {\textbf{ each} $y_i \in Y$}
			 {
				$\cbar{i}(N) :=  \cba{i}(c_1) \vee \cba{i}(c_2)\ldots \vee \cba{i}(c_k)$;~~~~~~  	\tcc{ $\cbar{i}(N)\rightarrow \cba{i}(N)$} 
				$\cbbr{i}(N) :=  \cbb{i}(c_1) \vee \cbb{i}(c_2)\ldots \vee \cbb{i}(c_k)$;~~~~~~~~ 	\tcc{$\cbbr{i}(N) \rightarrow \cbb{i}(N)$}
			}	
                  $\left(\cbva{N}, \cbvb{N}\right)= \textsc{Perform\_Cegar}(N, (\cbar{i}(N), \cbbr{i}(N))_{y_i\in Y})$\;
		}
	}
		
        \Return $\left(\cbva{N}, \cbvb{N}\right)$\;
\end{algorithm}

\subsection{A parallel synthesis algorithm}
\label{sec:booalgo}
The DAG representation of $\varphi(X,Y)$ gives a natural, recursive
decomposition of the specification, and also defines a partial order
of dependencies between the corresponding synthesis sub-problems.
Algorithm {\textsc{Compute}} can be invoked in parallel on nodes in
the DAG that are not ordered w.r.t. this partial order, as long as
{\textsc{Compute}} has already been invoked on their children.  This
suggests a simple parallel approach to Boolean functional synthesis.
Algorithm {\textsc{ParSyn}}, shown below, implements this approach,
and is motivated by a message-passing architecture.  We consider a
standard manager-worker configuration, where one out of available $m$
cores acts as the manager, and the remaining $m-1$ cores act as
workers.  All communication between the manager and workers is assumed
to happen through explicit {\bfseries send} and {\bfseries receive}
primitives.
\begin{algorithm}[h!]
\caption{{\textsc{ParSyn}}}
 \label{alg:cparcegar}
  \KwIn{AND-OR DAG with root $Rt$ representing $\varphi(X,Y)$ in NNF form}
  \KwOut{$(g_1, \ldots g_n)$ that effectively realize $Y$ in $\varphi(X,Y)$ \\\nonl \hrulefill}

\tcc{Algorithm for Manager }

Queue $Q$ \; \tcc{Invariant: Q has nodes that can be processed in parallel, i.e., leaves or nodes whose children have their $\cbvasimp$, $\cbvbsimp$ computed.}

Insert all leaves of DAG into $Q$\;

\While {all DAG nodes not processed}
{

		\While {a worker $W$ is idle and $Q$ is not empty}
		{
			Node $N := Q$.front()\;
			{\bfseries send} node $N$ for processing to $W$\; 

			{\bfseries if} $N$ has children $c_1, \ldots c_k$ {\bfseries then} {\bfseries send} $\cbva{c_j},\cbvb{c_j}$ for $1\le j \le k$ to $W$\;                        
		} 

		{\bf wait until} some worker $W'$ processing node $N'$ becomes free\;
                {\bfseries receive} $\left(\cbvasimp, \cbvbsimp\right)$ from
                $W'$, and store as $\left(\cbva{N'}, \cbvb{N'}\right)$\;
                Mark node $N'$ as processed\;
		
		\For {{\bf each} parent node $N''$ of $N'$} 
		{ 
			 {\bfseries if} all children of $N''$  are processed {\bfseries then}
			  insert $N''$ into $Q$
		}
               
} \tcc{All DAG nodes are processed; return $\neg\cbvbsimp$ or
  $\cbvasimp$ from root $Rt$}

	\Return $\left(\neg\cbb{1}(Rt), \ldots \neg\cbb{n}(Rt)\right)$ \hfill \tcp{or alternatively $\left(\cba{1}(Rt), \ldots \cba{n}(Rt)\right)$} 
{\nonl \hrulefill}
        
        \tcc{Algorithm for Worker $W$}
{\bfseries receive} node $N$ to process, and $\cbva{c_j}$, $\cbvb{c_j}$ for every child $c_j$ of $N$, if any\;
$\left(\cbvasimp, \cbvbsimp\right)$ := \textsc{Compute}($N$) \;
{\bfseries send}  $\left(\cbvasimp, \cbvbsimp\right)$ to Manager \;
\end{algorithm}

The manager uses a queue $Q$ of ready-to-process nodes. 
Initially, $Q$ is initialized with the leaf nodes in the DAG, and we
maintain the invariant that all nodes in $Q$ can be processed in
parallel.  If there is an idle worker $W$ and if $Q$ is not empty, the
manager assigns the node $N$ at the front of $Q$ to worker $W$ for
processing.  If $N$ is an internal DAG node, the manager also sends
$\cbva{c_j}$ and $\cbvb{c_j}$ for every child $c_j$ of $N$ to $W$. If
there are no idle workers or if $Q$ is empty, the manager waits for a
worker, say $W'$, to finish processing its assigned node, say $N'$.
When this happens, the manager stores the result sent by $W'$ as
$\cbva{N'}$ and $\cbvb{N'}$.  It then inserts one or more parents
$N''$ of $N'$ in the queue $Q$, if all children of $N''$ have been processed. The above steps are
repeatedly executed at the manager until all DAG nodes have been
processed.  The job of a worker $W$ is relatively simple: on being
assigned a node $N$, and on receiving $\cbva{c_j}$ and $\cbvb{c_j}$
for all children $c_j$ of $N$, it simply executes Algorithm
{\textsc{Compute}} on $N$ and returns $\left(\cbva{N},
\cbvb{N}\right)$.

Note that Algorithm {\textsc{ParSyn}} is guaranteed to progress as
long as all workers complete processing the nodes assigned to them in
finite time.  The partial order of dependencies between nodes ensures
that when all workers are idle, either all nodes have already been
processed, or at least one unprocessed node has $\cbvasimp$ and
$\cbvbsimp$ computed for all its children, if any.

\subsection{Extending the Composition Lemma and algorithms}
\label{sec:extensions}
So far, we have considered DAGs in which all internal nodes were
either AND- or OR-nodes. We now extend our results to more general
DAGs. We do this by generalizing the Composition Lemma to arbitrary Boolean
operators.  Specifically, given the refinements $\cbar{i}(c_j)$ and
$\cbbr{i}(c_j)$ at all children $c_j$ of a node $N$, we show how to
compose these to obtain $\cbar{i}(N)$ and $\cbbr{i}(N)$, when $N$ is
labeled by an arbitrary Boolean operator.  Note that the CEGAR
technique discussed in Section~\ref{sec:simple} can be used to
abstract the refinements $\cbar{i}$ and $\cbbr{i}$ to $\cba{i}$ and
$\cbb{i}$ respectively, at any node of interest.  Therefore, with our
generalized Composition Lemma, we can use compositional synthesis for
specifications represented by general DAGs, even without computing
$\cba{i}$ and $\cbb{i}$ exactly at all DAG nodes.  This gives an
extremely powerful approach for parallel, compositional synthesis.

Let $\formula{N} = \op(\formula{c_1}, \ldots \formula{c_r})$, where
$\op$ is an $r$-ary Boolean operator.  For convenience of notation, we
use $\neg N$ to denote $\neg \formula{N}$, and similarly for other
nodes, in the subsequent discussion.  Suppose we are given
$\cbar{i}(c_j)$, $\cbbr{i}(c_j)$, $\cbar{i}(\neg c_j)$ and
$\cbbr{i}(\neg c_j)$, for $1 \le j \le r$.  We wish to compose these
appropriately to compute $\cbar{i}(N)$, $\cbbr{i}(N)$,
$\cbar{i}(\neg N)$ and $\cbbr{i}(\neg N)$ for $1 \le i \le n$.  Once we
have these refinements, we can adapt
Algorithm~\ref{alg:composeatop} to work for node $N$, labeled by an arbitrary Boolean operator
$\op$.

\begin{figure}[h]
\begin{center}
  \scalebox{0.8}{
\begin{tikzpicture}[sibling distance=4em,
  main node/.style = {shape=circle,    draw, align=center},
  tri node/.style={shape=regular polygon, regular polygon sides=3, draw, scale=0.8},    
  edge from parent/.style={draw},
  no edge below/.style={    
        every child/.append style={solid,       
         edge from parent/.style={dashed}}},
    level 2/.style={level distance=7mm} 
]

  \node[main node] (1) {$\op$}
        child { node[minimum size=0] {}
              child {node[tri node]  {$c_1$}}}
        child { node {}
              child {node[tri node] (a) {$c_2$}}}
        child[missing] { node {}}
        child {node {}
              child {node[tri node] (b){$c_r$}}};  

 
 \node[main node] (2) [right of=1, node distance=5.5cm] {$\op$}
        child { node {$z_1$}[no edge below]
                     child { 
                     node {}
                           child {node[tri node]  {$c_1$}}}}
        child { node (c) {$z_2$}[no edge below]
                     child {
node {}
                           child {node[tri node] {$c_2$}}}}
        child[missing] { node {}}
        child { node (d) {$z_r$}[no edge below]
                     child {
node {}
                           child {node[tri node] {$c_r$}}}}
;

\path(a) -- (b) node[midway]{$\cdots$};  
\path(c) -- (d) node[midway]{$\cdots$};
\end{tikzpicture}
}
\caption{An $\op$ formula with $r$ variables (left) and its decomposition (right)}
\label{fig:decomp2}
\end{center}
\end{figure}
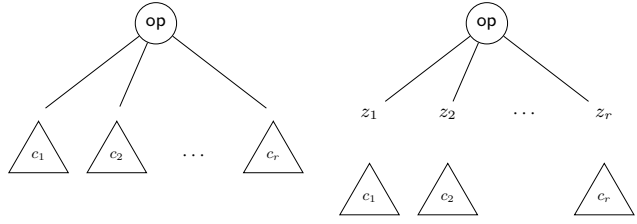

To understand how composition works for $\op$, consider the formula
$\op(z_1, \ldots z_r)$, where $z_1, \ldots z_r$ are fresh Boolean
variables, as shown in Figure \ref{fig:decomp2}.
Clearly, $\formula{N}$ can be viewed as $(\cdots(\op(z_1, \ldots
z_r)[\subst{z_1}{\formula{c_1}}])
\cdots)[\subst{z_r}{\formula{c_r}}]$.  For simplicity of notation, we
write $\op$ instead of $\op(z_1,\ldots,z_r)$ in the following
discussion.  W.l.o.g., let $z_1 \prec z_2 \prec \cdots \prec z_r$ be a
total ordering of the variables $\{z_1, \ldots z_r\}$.  Given $\prec$,
suppose we compute the formulas $\cbar{z_l}(\op)$, $\cbbr{z_l}(\op)$,
$\cbar{z_l}(\neg\op)$ and $\cbbr{z_l}(\neg\op)$ in negation normal
form (NNF), for all $l \in \{1, \ldots r\}$.  Note that these formulas
have support $\{z_{l+1}, \ldots z_r\}$, and do not have variables in
$X \cup Y$ in their support.  We wish to ask if we can compose these
formulas with $\cbar{i}(c_j)$, $\cbbr{i}(c_j)$, $\cbar{i}(\neg c_j)$
and $\cbbr{i}(\neg c_j)$ for $1 \le j \le r$ to compute $\cbar{i}(N)$,
$\cbbr{i}(N)$, $\cbar{i}(\neg N)$ and $\cbbr{i}(\neg N)$, for all $i
\in \{1, \ldots n\}$.  It turns out that we can do this.

Recall that in NNF, negations appear (if at all) only on literals.
Let $\Upsilon_{l, \op}$ be the formula obtained by replacing every
literal $\neg z_s$ in the NNF of $\cbbr{z_l}(\op)$ with a fresh
variable $\overline{z_s}$.  Similarly, let $\Omega_{l, \op}$ be
obtained by replacing every literal $\neg z_s$ in the NNF of
$\cbar{z_l}(\op)$ with the fresh variable $\overline{z_s}$.  The
definitions of $\Upsilon_{l, \neg\op}$ and $\Omega_{l, \neg\op}$ are
similar.  Replacing $\neg z_s$ by a fresh variable $\overline{z_s}$
allows us to treat the literals $z_s$ and $\neg z_s$ independently in
the NNF of $\cbbr{z_l}(\op)$ and $\cbar{z_l}(\op)$.  The ability to
treat these independently turns out to be important when formulating
the generalized Composition Lemma.
Let $\left(\Upsilon_{l,\op}\left[\subst{z_s}{\cbar{i}(\neg
    c_{s})}\right]\left[\subst{\overline{z_s}}{\cbar{i}(c_{s})}\right]\right)_{s=l+1}^r$
denote the formula obtained by substituting $\cbar{i}(\neg c_s)$ for
$z_s$ and $\cbar{i}(c_s)$ for $\overline{z_s}$, for every $s \in
\{l+1, \ldots r\}$, in $\Upsilon_{l,\op}$.  The interpretation of
$\left(\Omega_{l,\op}\left[\subst{z_s}{\cbar{i}(\neg
    c_{s})}\right]\left[\subst{\overline{z_s}}{\cbar{i}(c_{s})}\right]\right)_{s=l+1}^r$
is analogous.  Our generalized Composition Lemma can now be stated as follows. 
\begin{lemma}[Generalized Composition Lemma] 
\label{prop:gen}
Let $\formula{N} = \op(\formula{c_1}, \ldots \formula{c_r})$, where $\op$
is an $r$-ary Boolean operator. For each $1\leq i\leq n$ and $1\leq \ell\leq
r$:
\begin{align*}
1.~&\cbar{i}(c_l) \wedge   
\left(\Omega_{l,\op}\left[\subst{z_s}{\cbar{i}(\neg c_{s})}\right]\left[\subst{\overline{z_s}}{\cbar{i}(c_{s})}\right]\right)_{s=l+1}^r ~\rightarrow~
\cba{i}(N) \nonumber\\
2.~&\cbar{i}(\neg c_l) \wedge  
\left(\Upsilon_{l,\op}\left[\subst{z_s}{\cbar{i}(\neg c_{s})}\right]\left[\subst{\overline{z_s}}{\cbar{i}(c_{s})}\right]\right)_{s=l+1}^r
~\rightarrow~  \cba{i}(N) \nonumber\\
3.~&\cbbr{i}(c_l) \wedge   
\left(\Omega_{l,\op}\left[\subst{z_s}{\cbbr{i}(\neg c_{s})}\right]\left[\subst{\overline{z_s}}{\cbbr{i}(c_{s})}\right]\right)_{s=l+1}^r ~\rightarrow~
\cbb{i}(N) \nonumber\\
4.~&\cbbr{i}(\neg c_l) \wedge 
\left(\Upsilon_{l,\op}\left[\subst{z_s}{\cbbr{i}(\neg c_{s})}\right]\left[\subst{\overline{z_s}}{\cbbr{i}(c_{s})}\right]\right)_{s=l+1}^r
~\rightarrow~  \cbb{i}(N) \nonumber
\end{align*}
If we replace $\op$ by $\neg\op$ above, we get refinements
of $\cba{i}(\neg N)$ and $\cbb{i}(\neg N)$.
\end{lemma}
\begin{proof}
We provide a proof for the first implication.  The proofs for the other implications are similar.
Consider an assignment $\eta$ of values to $X \cup \{y_1, \ldots
y_n\}$ such that $\eta$ satisfies the left hand side of implication
(1).  We show below that $\eta$ satisfies $\cba{i}(N)$ as
well.

Let $\eta^\star$ denote an assignment of values to variables that
coincides with $\eta$ for all variables, except possibly for $y_i$.
Formally, $\eta^\star(v) = \eta(v)$ for $v \in X \cup \{y_1, \ldots
y_{i-1}, y_{i+1}, \ldots y_n\}$ and $\eta^\star(y_i) = 0$. We use
$\eta^\star(v)$ to denote the value assigned to variable $v$ in
$\eta^\star$.  If $\psi$ is a Boolean formula, we abuse notation and
use $\eta^\star(\psi)$ to denote the value that $\psi$ evaluates to,
when variables are assigned values according to $\eta^\star$.

Since the right hand side of implication (1) does not have
$y_i$ in its support, it suffices to show that $\eta^\star$ satisfies
$\cba{i}(N)$.  Furthermore, since neither side of implication (1)
depends on $\{y_1, \ldots y_{i-1}\}$, our arguments work for all
values of $\{y_1, \ldots y_{i-1}\}$.  Hence, it suffices to show that
$\eta^\star(N) = 0$.

\begin{claim}[1] $\eta^\star(c_l) = 0$ 
\end{claim}
\begin{proof}
To see why this is true, note that $\eta^\star$ satisfies the left
hand side of implication (1), and hence it satisfies
$\cbar{i}(c_l)$.  Since $\eta^\star(y_i) = 0$, it follows
from the definition of $\cbar{i}(\cdot)$ that
$\eta^\star(c_l) = 0$.
\end{proof}

Note that $\eta^\star$ also satisfies
$\left(\Omega_{l,\op,\prec}\left[\subst{z_s}{\cbar{i}(\neg c_{s})}\right]\left[\subst{\overline{z_s}}{\cbar{i}(c_{s})}\right]\right)_{s=l+1}^r$.
Define $\rho$ to be an assignment of values to
$\{z_{l+1}, \overline{z_{l+1}}, \ldots z_r, \overline{z_r}\}$ such
that $\rho(z_s) = \eta^\star\left(\cbar{i}(\neg c_s)\right)$
and $\rho(\overline{z_s})
= \eta^\star\left(\cbar{i}(c_s)\right)$, for all
$s \in \{l+1, \ldots r\}$.  It follows from the definition above that
$\rho$ is a satisfying assignment of $\Omega_{l,\op,\prec}$.

\begin{claim}[2]
For every $s \in \{l+1, \ldots r\}$,  either $\rho(z_s) = 0$ or $\rho(\overline{z_s}) = 0$. Further, for every $s \in \{l+1, \ldots r\}$,  if $\rho(z_s) = 1$, then $\eta^\star(c_s) = 1$, and if $\rho(\overline{z_s}) = 1$, then $\eta^\star(c_s) = 0$.
\end{claim}
\begin{proof}
The proof of the first statement is by contradiction.  If possible, let $\rho(z_s)= 1$ and $\rho(\overline{z_s}) = 1$ for some $s \in \{l+1, \ldots r\}$.  By definition of $\rho$, we have $\eta^\star\left(\cbar{i}(\neg c_s)\right) = 1$ and
$\eta^\star\left(\cbar{i}(c_s)\right) = 1$.  By definition
of $\cbar{i}(\cdot)$, it follows that both $\neg c_s$ and
$c_s$ evaluate to $0$ for the assignment $\eta^\star$ -- a
contradiction!


For the second statement, note that by definition, if $\rho(z_s) = 1$,
then $\eta^\star\left(\cbar{i}(\neg c_s)\right) = 1$.  Since
$\eta^\star(y_i) = 0$, it follows from the definition of
$\cbar{i}(\cdot)$ that $\eta^\star(\neg(c_s) = 0$.
Equivalently, we have $\eta^\star(c_s) = 1$.  Similarly, if
$\rho(\overline{z_s}) = 1$, then by definition,
$\eta^\star\left(\cbar{i}(c_s)\right) = 1$, and hence
$\eta^\star(c_s) = 0$.
\end{proof}

Finally, we define an assignment $\widehat{\rho}$ of values to
$\{z_{l+1}, \overline{z_{l+1}}, \ldots z_r, \overline{z_r}\}$ as
follows: for all $s \in \{l+1, \ldots r\}$, $\widehat{\rho}(z_s)
= \rho(z_s)$ if either $\rho(z_s) = 1$ or $\rho(\overline{z_s}) = 1$,
and $\widehat{\rho}(z_s) = \eta^\star(c_s)$ otherwise;
$\rho(\overline{z_s})$ is always equal to $\neg \rho(z_s)$.  From
Claim~(2), we can now infer that $\widehat{\rho}(z_s)
= \eta^\star(c_s)$ if either $\rho(z_s) = 1$ or
$\rho(\overline{z_s}) = 1$.  Therefore, we obtain the following claim,

\begin{claim}[3]
$(\widehat{\rho}(z_s), \widehat{\rho}(\overline{z_s}) =$
$\left(\eta^\star(c_s), \neg \eta^\star(c_s)\right)$ for
all $s \in \{l+1, \ldots r\}$.
\end{claim}

With the above claims, we can now prove the first implication/statement of the lemma. From Claim~(1), the values of $(\rho(z_s), \rho(\overline{z_s}))$
are either $(0,1), (1,0)$ or $(0,0)$, for all $s \in \{l+1,\ldots
r\}$. Therefore, $\rho(v) \rightarrow \widehat{\rho}(v)$ for all
$v \in \{z_{l+1}, \overline{z_{l+1}}, \ldots z_r, \overline{z_r}\}$.
Furthermore, since $\Omega_{l,\op,\prec}$ is obtained by replacing all
literals $\neg z_s$ with $\overline{z_s}$ in the NNF of
$\cbar{z_l}(\op)$, $\Omega_{l,\op}$ is positive unate in
$\{z_{l+1}, \overline{z_{l+1}}, \ldots z_r, \overline{z_r}\}$.  It
follows that since $\rho$ satisfies $\Omega_{l,\op,\prec}$ and
$\rho(v) \rightarrow \widehat{\rho}(v)$ for all $v$, $\widehat{\rho}$
also satisfies $\Omega_{l,\op,\prec}$.  

From the definition of $\Omega_{l,\op,\prec}$, it is easy to see that
$\left(\Omega_{l,\op,\prec}[\subst{\overline{z_s}}{\neg
z_s}]\right)_{s=l+1}^r$ is exactly $\cba{z_l}(\op)$.  Therefore,
$\Omega_{l,\op,\prec}$ evaluated at $\widehat{\rho}$ has the same value,
i.e. $1$, as $\cba{z_l}(\op)$ evaluated at $\widehat{\rho}$. From
the definition of $\cba{z_l}(\op)$, it follows that $\op(z_1, \ldots
z_r)$ evaluates to $0$ for every assignment $\zeta$ of values to
$\{z_1, \ldots z_r\}$ such that $\zeta(z_l) = 0$, and $\zeta(z_s)
= \widehat{\rho}(z_s)$ for all $s \in \{l+1, \ldots r\}$.

Now, Claims (1) and (3) imply that $\zeta(z_s) = \eta^\star(c_s)$ for all $s \in \{l, \ldots r\}$.  Therefore, $\eta^\star\left(\op[\subst{z_1}{\formula{c_1}}] \cdots
[\subst{z_r}{\formula{c_r}}]\right) = 0$.  Since $\formula{N}
= \op[\subst{z_1}{\formula{c_1}}] \cdots [\subst{z_r}{\formula{c_r}}]$,
we have $\eta^\star(N) = 0$.  This proves the first statement/implication  of the lemma.

The other implications can be similarly proved following this same template.\qed
\end{proof}

We simply illustrate the idea behind the lemma
with an example here. Suppose $\formula{N} = \formula{c_1} \wedge \neg
\formula{c_2} \wedge (\neg \formula{c_3} \vee \formula{c_4})$, where
each $\formula{c_j}$ is a Boolean function with support $X \cup \{y_1,
\ldots y_n\}$.  We wish to compute a refinement of $\cba{i}(N)$, using
refinements of $\cba{i}(c_j)$ and $\cba{i}(\neg c_j)$ for $j \in
\{1,\ldots 4\}$.
Representing $N$ as $\op(c_1, c_2, c_3, c_4)$, let $z_1, \ldots z_4$
be fresh Boolean variables, not in $X \cup \{y_1, \ldots y_n\}$; then
$\op(z_1, z_2, z_3, z_4) = z_1 \wedge \neg z_2 \wedge (\neg z_3 \vee
z_4)$.  For ease of exposition, assume the ordering $z_1 \prec z_{2}
\prec z_{3} \prec z_{4}$.  By definition, $\cba{z_2}(\op) = \left(\neg
\exists z_1\, (z_1 \wedge \neg z_2 \wedge (\neg z_3 \vee
z_4))\right)[\subst{z_2}{0}]$ $=$ $z_3 \wedge \neg z_4$, and suppose
$\cbar{z_2}(\op) = \cba{z_2}(\op)$. Replacing $\neg z_{4}$ by
$\overline{z_{4}}$, we then get
$\Omega_{2,\op}=z_{3}\wedge\overline{z_{4}}$.

Recalling the definition of $\cbar{z_2}(\cdot)$, if we set $z_{3}=1$,
$z_{4}=0$ and $z_2=0$, then $\op$ must evaluate to $0$ regardless of
the value of $z_1$.  By substituting $\cbar{i}(\neg c_{3})$ for
$z_{3}$ and $\cbar{i}(c_{4})$ for $\overline{z_{4}}$ in
$\Omega_{2,\op}$, we get the formula $\cbar{i}(\neg c_{3}) \wedge
\cbar{i}(c_{4})$.  Denote this formula by $\chi$ and note that its
support is $X \cup \{y_{i+1}, \ldots y_n\}$.  Note also from the
definition of $\cbar{i}(\cdot)$ that if $\chi$ evaluates to $1$ for
some assignment of values to $X \cup \{y_{i+1}, \ldots y_n\}$ and if
$y_i = 0$, 
evaluates to $0$ and $\formula{c_{4}}$ evaluates to $0$, regardless of
the values of $y_1, \ldots y_{i-1}$.  This means that $z_{3} = 1$ and
$z_{4} = 0$, and hence $\cbar{z_2}(\op) = 1$.  If $z_2$ (or
$\formula{c_2}$ can also be made to evaluate to $0$ for the same
assignment of values to $X \cup \{y_i, y_{i+1}, \ldots y_n\}$, then $N
= \op(c_1, \ldots c_r)$ must evaluate to $0$, regardless of the values
of $\{y_1, \ldots y_{i-1}\}$.  Since $y_i = 0$, 
values assigned to $X \cup \{y_{i+1}, \ldots y_n\}$ must therefore be
a satisfying assignment of $\cba{i}(N)$.  One way of achieving this is
to ensure that $\cba{i}(c_2)$ evaluates to $1$ for the same assignment
of values to $X \cup \{ y_{i+1}, \ldots y_n\}$ that satisfies $\chi$.
Therefore, we require the assignment of values to $X \cup \{ y_{i+1}, \ldots
y_n\}$ to satisfy $\chi \wedge \cba{i}(c_2)$, or even $\chi \wedge
\cbar{i}(c_2)$.  Since $\chi = \cbar{i}(\neg c_{3}) \wedge
\cbar{i}(c_{4})$, we get $\cbar{i}(c_2) \wedge {\cbar{i}(\neg c_{3})}
\wedge {\cbar{i}(c_{4})}$ as a refinement of $\cba{i}(N)$.

\paragraph{Applying the generalized Composition Lemma:} Lemma~\ref{prop:gen}
suggests a way of compositionally obtaining $\cbar{i}(N)$,
$\cbbr{i}(N)$, $\cbar{i}(\neg N)$ and $\cbbr{i}(\neg N)$ for an
arbitrary Boolean operator $\op$.  Specifically, the disjunction of
the left-hand sides of implications (1) and (2) in
Lemma~\ref{prop:gen}, disjoined over all $l \in \{1, \ldots r\}$ and
over all total orders ($\prec$) of $\{z_1, \ldots z_r\}$, gives a refinement of
$\cba{i}(N)$.  A similar disjunction of the left-hand sides of
implications (3) and (4) in Lemma~\ref{prop:gen} gives a refinement of
$\cbb{i}(\varphi)$. The cases of $\cba{i}(\neg N)$ and $\cbb{i}(\neg
N)$ are similar.  This suggests that for each operator $\op$ that
appears as label of an internal DAG node, we can pre-compute a
template of how to compose $\cbar{i}$ and $\cbbr{i}$ at the children of
the node to obtain $\cbar{i}$ and $\cbbr{i}$ at the node itself.  In
fact, pre-computing this template for $\op = \vee$ and $\op = \wedge$
by disjoining as suggested above, gives us exactly the left-to-right
implications, i.e., refinements of $\cba{i}(N)$ and $\cbb{i}(N)$, as
given by Lemma~\ref{prop:comp}.  We present   
templates for some other common Boolean operators like \emph{if-then-else} in the next subsection.

Once we have pre-computed templates for composing $\cbar{i}$ and
$\cbbr{i}$ at children of a node $N$ to get $\cbar{i}(N)$ and
$\cbbr{i}(N)$, we can use these pre-computed templates in
Algorithm~\ref{alg:composeatop}, just as we did for AND-nodes.  This
allows us to apply compositional synthesis on general DAG
representations of Boolean relational specifications.

\paragraph{Optimizations using partial computations:}

Given $\cbar{i}$ and $\cbbr{i}$ at children of a node $N$, we have
shown above how to compute $\cbar{i}(N)$ and $\cbbr{i}(N)$.  To
compute $\cba{i}(N)$ and $\cbb{i}(N)$ exactly, we can use the CEGAR
technique outlined in Section~\ref{sec:simple}.  While this is
necessary at the root of the DAG, we need not compute $\cba{i}(N)$ and
$\cbb{i}(N)$ exactly at each intermediate node.  In fact, the
generalized Composition Lemma allows us to proceed with $\cbar{i}(N)$
and $\cbbr{i}(N)$.  This suggests some optimizations: (i) Instead of
using the error formulas introduced in Section~\ref{sec:simple}, that
allow us to obtain $\cba{i}(N)$ and $\cbb{i}(N)$ exactly, we can use
the error formula used in~\cite{fmcad2015:skolem}.  The error formula
of~\cite{fmcad2015:skolem} allows us to obtain some Skolem function
for $y_i$ (not necessarily $\cba{i}(N)$ or $\neg\cbb{i}(N)$) using the
sub-specification $\formula{N}$ corresponding to node $N$.  We have
found CEGAR based on this error formula to be more efficient in
practice, while yielding refinements of $\cba{i}(N)$ and $\cbb{i}(N)$.
In fact, we use this error formula in our implementation. (ii) We
can introduce a \emph{timeout} parameter, such that
$\cbva{N},\cbvb{N}$ are computed exactly at each internal node until
we timeout happens. Subsequently, for the nodes still under process,
we can simply combine $\cbar{i}$ and $\cbbr{i}$ at their children
using our pre-computed composition templates, and not invoke CEGAR at
all.  The only exception to this is at the root node of the DAG where
CEGAR must be invoked.

\subsection{Application of Lemma \ref{prop:gen} through Examples}
In this subsection, with the help of examples, we present the computation of the $\cba{.}$ and $\cbb{.}$ templates for various boolean formulae. We also compare these templates with those obtained by using Algorithm \ref{alg:composeatop}. 
\begin{example}
To begin with, we demonstrate how to derive the templates for the $\wedge$ operator. 

Let $\varphi = c_1 \wedge c_2$ where $c_1, c_2$ are arbitrary boolean formulae and have (a subset of) $X \cup Y$ in their support. As shown in Figure \ref{fig:decomp2}, let $z_1$ and $z_2$ be fresh Boolean variables. We  first compute the relevant $\cbbr{}(.)$ and $\cbar{}(.)$'s for $z_1 \wedge z_2$ and then use Lemma \ref{prop:gen} to compute the template for $\cba{i}(\varphi)$ and $\cbb{i}(\varphi)$.

Let $F = z_1 \wedge z_2$.  Without loss of generality, let the ordering be : $z_1 \prec z_{2} $. Given an ordering and boolean function $F$, we compute the $\cba{}(.)$ and $\cbb{}(.)$ sets for each variable as follows: $F[ \subst{z_1}{1}] = z_2$; $F[\subst{z_1}{0}] = 0$

Therefore, by definition~\ref{def:cb0-cb1}, $\cba{1}(F) = 1$ and $\cbb{1}(F) = \neg z_2$. On existentially quantifying $z_1$ from $F$, we get:
$(\exists z_1 F) =  (z_2 \vee 0) = z_2$, i.e.,  $\exists z_1 F[ \subst{z_2}{1}] = 1$; $F[\subst{z_2}{0}] = 0$. Therefore, we have $\cba{2}(F) = 1$ and $\cbb{2}(F) = 0$.

Using  the Generalized Compositional lemma, (Lemma \ref{prop:gen}), we get:
\begin{equation*}
\begin{aligned}
\cbar{i}(c_1)  \wedge 1 \rightarrow~ \cba{i}(\varphi)  ;  \\
\cbar{i}(\neg c_1) \wedge 0 \rightarrow~ \cba{i}(\varphi)  ;  \\
\cbar{i}( c_2) \wedge 1 \rightarrow~ \cba{i}(\varphi) ; \\
\cbar{i}( \neg c_2) \wedge 0 \rightarrow~ \cba{i}(\varphi) ; 
\end{aligned}
\end{equation*}

On disjunction of the terms on the LHS, we get the following template 

\(\cbar{i}(c_1) \wedge \cbar{i}(c_2) ~\rightarrow~ \cba{i}(\varphi) \) ;

which is exactly the same as that proved in Lemma \ref{prop:comp}.

\end{example}
\begin{example}

We now consider the $\vee$ operator. Let $\varphi = c_1 \vee c_2$ where $c_1, c_2$ are arbitrary boolean formulae and have (a subset of) $X \cup Y$ in their support.
As beforelet $z_1$ and $z_2$ be fresh Boolean variables. 
Let $F = z_1 \vee z_2$.  Without loss of generality, let the ordering be : $z_1 \prec z_{2} $.

Then, $F[ \subst{z_1}{1}] = 1$; $F[\subst{z_1}{0}] = z_2$. Therefore, by definition~\ref{def:cb0-cb1}, $\cba{1}(F) = \neg z_2$ and $\cbb{1}(F) = 0$. On existentially quantifying $z_1$ from $F$, we get: $\exists z_1 F =  z_2 \vee 1 = 1$, and so $\cba{2}(F) = 0$ and $\cbb{2}(F) = 0$.

Again, using the Generalized Compositional lemma, (Lemma \ref{prop:gen}), we get:
\begin{equation*}
\begin{aligned}
\cbar{i}(c_1)  \wedge \cbar{i}(c_2) \rightarrow~ \cba{i}(\varphi)  ;  \\
\cbar{i}(\neg c_1) \wedge 0 \rightarrow~ \cba{i}(\varphi)  ;  \\
\cbar{i}( c_2) \wedge 0 \rightarrow~ \cba{i}(\varphi) ; \\
\cbar{i}( \neg c_2) \wedge 0 \rightarrow~ \cba{i}(\varphi) ; 
\end{aligned}
\end{equation*}

\end{example}

By disjuncting the terms in the LHS, we get the following template 

\(\cbar{i}(c_1) \wedge \cbar{i}(c_2) ~\rightarrow~ \cba{i}(\varphi) \) ;

which is, again, the same as that proved in Lemma \ref{prop:comp}.

\begin{example}
Now consider the \emph{if-then-else} or the $\ite$ operator. Let $\varphi = \ite(c_1, c_2, c_3)$ where $c_1, c_2, c_3$ are arbitrary boolean formulae and have (a subset of) $X \cup Y$ in their support. As before, let $z_1$, $z_2$ and $z_3$ be fresh Boolean variables.  We  first compute the relevant $\cbbr{}(.)$ and $\cbar{}(.)$'s for $\ite(z_1, z_2, z_3)$ and then use Lemma \ref{prop:gen} to compute the template for $\cba{i}(\varphi)$ and $\cbb{i}(\varphi)$.

Let $F =  \ite(z_1, z_2, z_3)$. This means that if $z_1$ evaluates to true, then $F = z_2$ else $F = z_3$. Let the ordering be : $z_1 \prec z_{2} \prec z_{3}$.

Then, $F[ \subst{z_1}{1}] = z_2$; $F[\subst{z_1}{0}] = z_3$ and by definition, $\cba{1}(F) = \neg z_3$ and $\cbb{1}(F) = \neg z_2$. On existentially quantifying $z_1$ from $F$, we get: $\exists z_1 F = z_2 \vee z_3$. 

To compute $\cba{2}(F)$ and $\cbb{2}(F)$, $\exists z_1 F [\subst{z_2}{1}]  = 1 ;  \exists z_1 F[ \subst{z_2}{0}] = z_3$, and so $\cba{2}(F) = \neg z_3$ and $\cbb{2}(F) = 0$.

Finally, on existentially quantifying $z_1$ and $z_2$ from $F$, we get:
$\exists z_1 \exists z_2 F = 1$; which gives $\cba{3}(F) = \cbb{3}(F) = 0$

From Lemma \ref{prop:gen}, we have:
\begin{equation*}
\begin{aligned}
\cbar{i}(c_1) \wedge \cbar{i}(c_3) ~\rightarrow~ \cba{i}(\varphi) \\
\cbar{i}(\neg c_1) \wedge \cbar{i}(c_2) ~\rightarrow~ \cba{i}(\varphi) \\
\cbar{i}( c_2) \wedge \cbar{i}(c_3) ~\rightarrow~ \cba{i}(\varphi) 
\end{aligned}
\end{equation*}


Combining these terms, we get the template for $\cba{i}(\varphi)$ as:\\
$(\cbar{i}(c_1) \wedge \cbar{i}(c_3)) \vee (\cbar{i}(\neg c_1) \wedge \cbar{i}(c_2))  \vee ( \cbar{i}( c_2) \wedge \cbar{i}(c_3)) \rightarrow \cba{i}(\varphi)$ \\

Similarly, the template for $\cbb{i}(\varphi)$ is: \\
$(\cbbr{i}(c_1) \wedge \cbbr{i}(c_3)) \vee (\cbbr{i}(\neg c_1) \wedge \cbbr{i}(c_2)) \vee ( \cbbr{i}( c_2) \wedge \cbbr{i}(c_3)) \rightarrow \cbb{i}(\varphi)$ \\

Note that, we can also represent $\ite(c_1, c_2, c_3)$ as a formula G containing only AND, OR and NOT operators and derive the template directly using Lemma \ref{prop:comp}. That is, let $G = (c_1 \wedge c_2) \vee (\neg c_1 \wedge c_3)$. Using Algorithm \ref{alg:composeatop} (and not doing Step 14 of \textsc{Perform\_Cegar}), we get the following:

\(\cbar{i}(c_1) \vee \cbar{i}(c_2)) \wedge (\cbar{i}(\neg c_1) \vee \cbar{i}(c_3)) \rightarrow \cba{i}(\varphi) \)

On simplication we get,

$(\cbbr{i}(c_1) \wedge \cbbr{i}(c_3)) \vee (\cbbr{i}(\neg c_1) \wedge \cbbr{i}(c_2)) \vee ( \cbbr{i}( c_2) \wedge \cbbr{i}(c_3)) \rightarrow \cbb{i}(\varphi)$ 

Here again, the templates given by Lemma \ref{prop:comp} and Lemma \ref{prop:gen} are the same. However, in the next example we consider a boolean function where the two differ.
\end{example}

\begin{example}

With the help this example, we demonstrate that Lemma \ref{prop:gen} can give better underapproximations than the approach presented in Algorithm \ref{alg:composeatop}.

Let $\varphi = (c_1 \oplus c_2) \wedge (c_1 \oplus c_3)$ where $c_1, c_2, c_3$ are arbitrary boolean formulae and have (a subset of) $X \cup Y$ in their support. As before, let $z_1$, $z_2$ and $z_3$ be fresh Boolean variables. Let $F = (z_1 \oplus z_2) \wedge (z_1 \oplus z_3)$

Without loss of generality, let the ordering be : $z_1 \prec z_{2} \prec z_{3}$. Then, $F[ \subst{z_1}{1}] = \neg z_2 \wedge \neg z_3$; $F[\subst{z_1}{0}] = z_2 \wedge z_3$. By definition, $\cba{1}(F) = \neg z_2 \vee  \neg z_3$ and $\cbb{1}(F) = z_2 \vee z_3$.  On existentially quantifying $z_1$ from $F$, we get:
$\exists z_1 F = (\neg z_2 \wedge \neg z_3) \vee (z_2 \wedge z_3)$. To compute $\cba{2}(F)$ and $\cbb{2}(F)$, $\exists z_1 F [\subst{z_2}{1}]  = z_3 ;  \exists z_1 F[ \subst{z_2}{0}] =  \neg z_3$. Therefore, $\cba{2}(F) = z_3$ and $\cbb{2}(F) = \neg z_3$

Again, on existentially quantifying $z_1$ and $z_2$ from $F$, we get:
$\exists z_1 \exists z_2 F = 1$ and $\cba{3}(F) = \cbb{3}(F) = 0$.

Using the compositional lemma we get:

\begin{equation*}
\begin{aligned}
\cbar{i}(c_1) \wedge  (\cbar{i}(c_2) \vee \cbar{i}(c_3))~\rightarrow~ \cba{i}(\varphi)  \\
\cbar{i}(\neg c_1) \wedge (\cbar{i}(\neg c_2) \vee \cbar{i}( \neg c_3))~\rightarrow~ \cba{i}(\varphi)  \\
\cbar{i}( c_2) \wedge \cbar{i}(\neg c_3) ~\rightarrow~ \cba{i}(\varphi) \\
\cbar{i}( \neg c_2) \wedge \cbar{i}(c_3) ~\rightarrow~ \cba{i}(\varphi) \\
\end{aligned}
\end{equation*}

Disjunction of the terms on the LHS, allows us to get:

\( (\cbar{i}(c_1) \wedge  (\cbar{i}(c_2)) \vee  ( \cbar{i}(c_1) \wedge \cbar{i}(c_3)) \vee (\cbar{i}(\neg c_1) \wedge \cbar{i}(\neg c_2)) \vee ( \cbar{i}(\neg c_1) \wedge \cbar{i}(\neg c_3)) \\ \vee (\cbar{i}( c_2) \wedge \cbar{i}(\neg c_3) ) \vee (\cbar{i}( \neg c_2) \wedge \cbar{i}(c_3)) ~\rightarrow~ \cba{i}(\varphi) \) 

However, if we represent $\varphi = (c_1 \oplus c_2) \wedge (c_1 \oplus c_3)$ as a boolean formula containing AND's and OR's as  $\varphi = ((\neg c_1 \wedge c_2) \vee (c_1  \wedge \neg c_2)) \wedge (((\neg c_1 \wedge c_3) \vee (c_1  \wedge \neg c_3))$. Using Algorithm \ref{alg:composeatop} (and not doing Step 14 of \textsc{Perform\_Cegar}), we only get the following: \( (\cbar{i}(c_1) \wedge  \cbar{i}(c_2)) \vee (\cbar{i}(c_1) \wedge \cbar{i}(c_3)) \vee (\cbar{i}(\neg c_1)  \wedge \cbar{i}(\neg c_2)) \vee (\cbar{i}(\neg c_1) \wedge \cbar{i}(\neg c_3))~\rightarrow~ \cba{i}(\varphi) \);

Note that in addition to the terms above,  the template for $\cba{i}(\varphi)$ derived using Lemma \ref{prop:gen} also has the additional terms
\( (\cbar{i}( c_2) \wedge \cbar{i}(\neg c_3)) \vee  (\cbar{i}( \neg c_2) \wedge \cbar{i}(c_3)) \).  It can easily be seen that these two terms are necessary, if $c_2 \neq c_3$ then $\varphi$ will not evaluate to true.
This example shows Lemma \ref{prop:gen} can give better underapproximations than Lemma \ref{prop:comp} for complex boolean formulae.
\end{example}

\section{Experimental results} \label{results}
\label{sec:expt}

\subsubsection*{Experimental methodology.}
We have implemented Algorithm~\ref{alg:cparcegar} with the error
formula from~\cite{fmcad2015:skolem} used for CEGAR in
Algorithm~\ref{alg:composeatop} (in function
{\textsc{Perform\_Cegar}), as described at the end of
Section~\ref{sec:extensions}.  We call this implementation $\parcegar$
in this section, and compare it with the following algorithms/tools:
$(i)$ $\cegarsk$: This is based on the sequential algorithm for
conjunctive formulas, presented in~\cite {fmcad2015:skolem}.  For
non-conjunctive formulas, the algorithm in~\cite{fmcad2015:skolem},
and hence $\cegarsk$, reduces to~\cite{Jian,Trivedi}.  $(ii)$
$\rsynth$: The {\it RSynth} tool as described in \cite{rsynth}.
$(iii)$ $\bloqqer$: As prescribed in~\cite{bierre}, we first generate
special QRAT proofs using the preprocessing tool \textsf{bloqqer}, and
then generate Boolean function vectors from the proofs using the
$\qrat$ tool.

Our implementation of $\parcegar$, available online
at~\cite{tacas2017:benchmarks}, makes extensive use of the
ABC~\cite{abc-tool} library to represent and manipulate Boolean
functions as AIGs.  We also use the default SAT solver provided by
ABC, which is a variant of MiniSAT.  We present our evaluation on
three different kinds of \emph{benchmarks}.
\begin{enumerate} 
\item{\em Disjunctive Decomposition Benchmarks}: Similar to~\cite{fmcad2015:skolem}, these benchmarks were generated
by considering some of the larger sequential circuits in the HWMCC10 benchmark suite, and  formulating 
the problem of disjunctively decomposing each circuit into components  as a problem of synthesizing a vector of Boolean functions.  Each generated benchmark is of the form $\exists Y \varphi(X,Y)$ where $\exists X (\exists Y \varphi(X,Y))$ is $\true$. However, unlike \cite{fmcad2015:skolem}, where each benchmark (if not already a conjunction of factors) had to be converted into factored form using Tseitin encoding (which introduced additional variables), we have used these benchmarks without Tseitin encoding.
\item {\em Arithmetic Benchmarks}:
These benchmarks were taken from the work described in~\cite{rsynth}. Specifically, the benchmarks considered are {\it floor}, {\it ceiling}, {\it decomposition}, {\it equalization} and {\it intermediate} (see~\cite{rsynth} for details).
\item {\em Factorization Benchmarks}:
We considered the integer factorization problem for different bit-widths,
as discussed in Section \ref{sec:introduction}. 

\end{enumerate}
For each arithmetic and factorization benchmark, we first specified
the problem instance as an { SMT} formula and then used {\it
Boolector}~\cite{boolector} to generate the Boolean version of the
benchmark.  For each arithmetic benchmark, three variants were
generated by varying the bit-width of the arguments of arithmetic
operators; specifically, we considered bit-widths of $32$, $128$ and
$512$. Similarly, for the factorization benchmark, we generated four
variants, using $8$, $10$, $12$ and $16$ for the bit-width of the
product.  Further, as $\bloqqer$ requires the input to be
in \texttt{qdimacs} format and $\rsynth$ in \texttt{cnf} format, we
converted each benchmark into \texttt{qdimacs} and $\texttt{cnf}$
formats using Tseitin encoding~\cite{tseitin68}.  All benchmarks and
the procedure by which we generated them are detailed
in \cite{tacas2017:benchmarks}.

\paragraph{Variable ordering:} 
We used the same ordering of variables for all algorithms.  For each
benchmark, the variables are ordered such that the variable which
occurs in the transitive fan-in of the least number of nodes in the
AIG representation of the specification, appears at the top. For
$\rsynth$ this translated to an interleaving of most of the input and
output variables.

\paragraph{Machine details}: All experiments were performed on a message-passing cluster, where each node had 20 cores and $64$ GB main memory, each core being a 2.20 GHz Intel Xeon processor. 
The operating system was Cent OS 6.5.
For $\cegarsk$, $\bloqqer$, and $\rsynth$, a single core on the
cluster was used.  For all comparisons, $\parcegar$ was executed on
$4$ nodes using $5$ cores each, so that we had both intra-node and
inter-node communication.  
 The maximum time given for execution was 3600 seconds, i.e., 1 hour. We also
restricted the total amount of main memory (across all cores) to be 16GB.
The metric used to compare the different algorithms was the time taken to synthesize Boolean functions. 
\medskip

\noindent{\bf Results.}
Our benchmark suite consisted of
$27$ disjunctive decomposition benchmarks, $15$ arithmetic benchmarks and $4$ factorization benchmarks. These benchmarks are fairly comprehensive in  size i.e., the number of AIG nodes ($|SZ|$) in the benchmark, and the  number of variables ($|Y|$) for which Boolean functions are to be synthesized. 
 Amongst disjunctive decomposition benchmarks, $|SZ|$ varied from $1390$ to $58752$ and $|Y|$ varied from $21$ to  $205$. 
Amongst the arithmetic benchmarks, $|SZ|$ varied from  $442$ to $11253$ and $|Y|$ varied from $31$ to $1024$.
The factorization benchmarks are the smallest and the most complex of the benchmarks, with 
$|SZ|$ varying from $122$ to $502$ and $|Y|$ varying from $8$ to $16$.

We now present the performance of the various algorithms. On $4$ of
the $46$ benchmarks, none of the tools succeeded. Of these, $3$
belonged to the {\em intermediate} problem type in the arithmetic
benchmarks, and the fourth one was the $16$ bit factorization
benchmark.
\medskip

\begin{figure*}[t]
\centering
\begin{subfigure}{2.3in}
  \includegraphics[angle=-90,  scale=0.23] {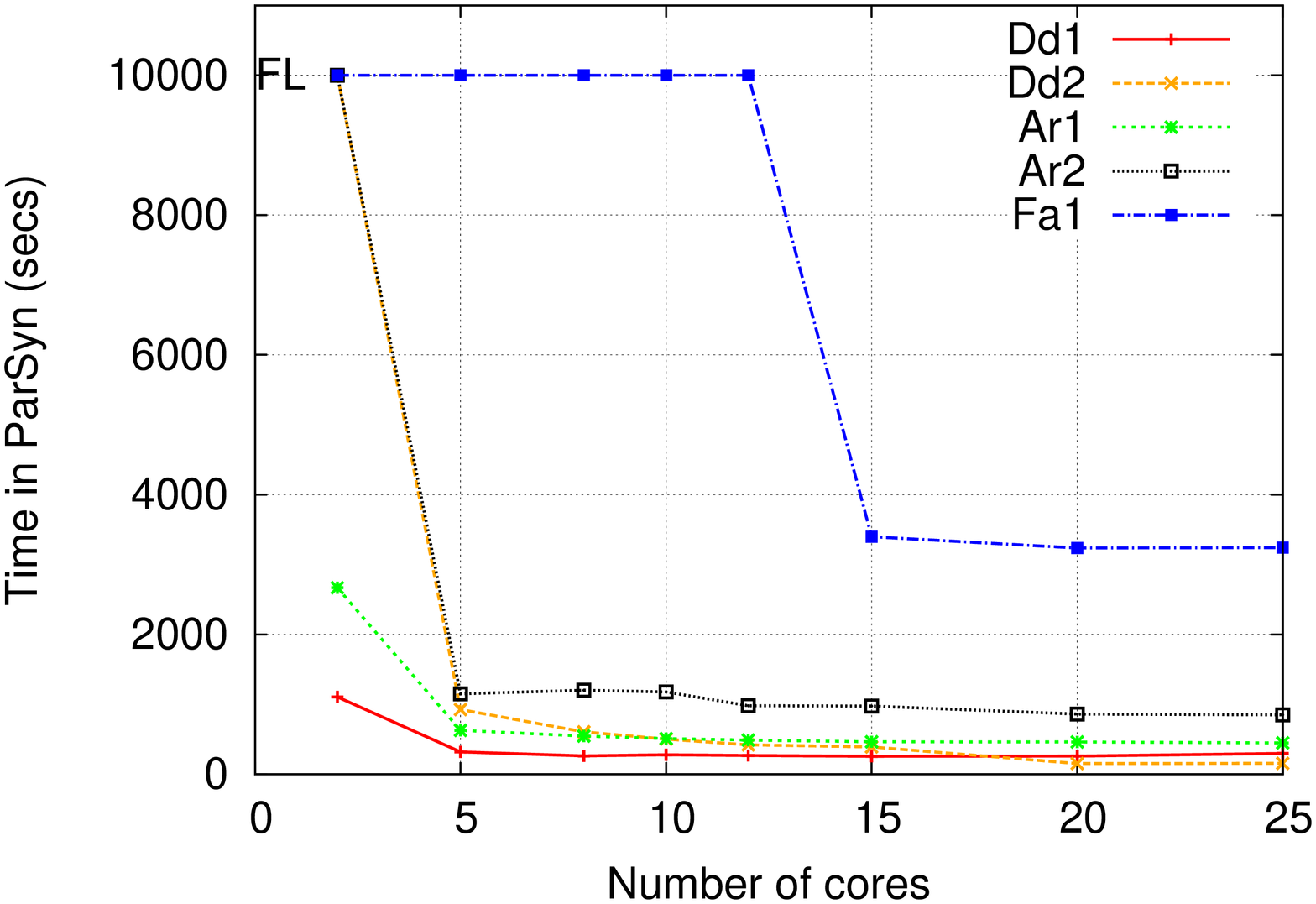} 
\caption{$\parcegar$ on different cores} 
\label{fig:parcegarcores}
\end{subfigure}
\begin{subfigure}{2.3in}
  \includegraphics[angle=-90,  scale=0.25]{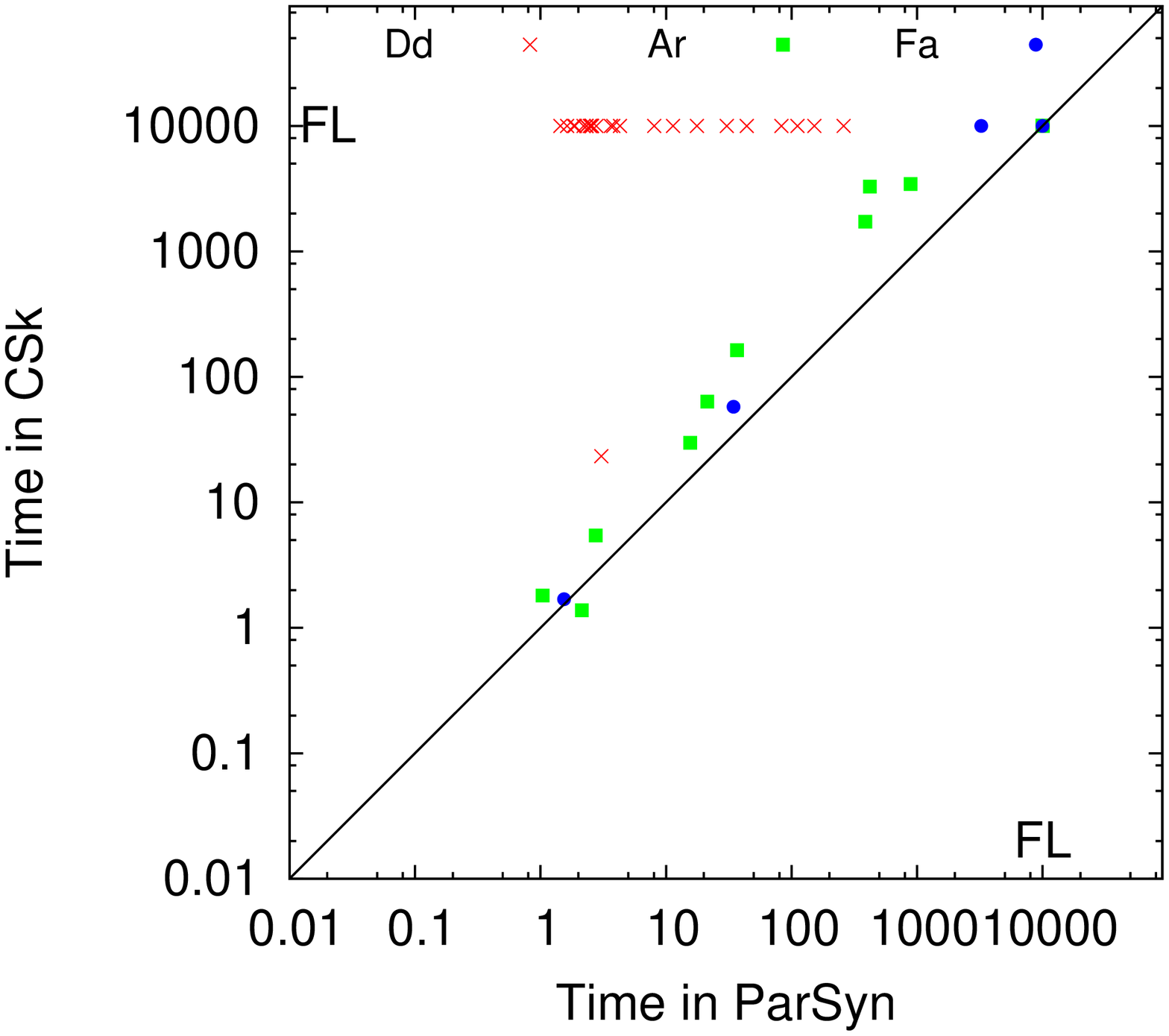}
  \caption{$\parcegar$ vs $\cegarsk$}
  \label{fig:parsk}
\end{subfigure}%
\caption{Legend: \texttt{Ar}: arithmetic, \texttt{Fa}: factorization, \texttt{Dd}: disjunctive decomposition. \texttt{FL}: benchmarks for which the corresponding algorithm was unsuccessful.}
\end{figure*}

\noindent {\it Effect of the number of cores}. 
For this experiment, we chose $5$ of the larger benchmarks.  Of these,
two benchmarks belonged to the disjunctive decomposition category, two
belonged to the arithmetic benchmark category and one was the 12 bit
factorization benchmark.  The number of cores was varied from $2$ to
$25$. With $2$ cores, $\parcegar$ behaves like a sequential algorithm
with one core acting as the manager and the other as the worker with
all computation happening at the worker core. Hence, with $2$ cores,
we see the effect of compositionality without parallelism. For number
of cores $>$ 2, the number of worker cores increase, and the
computation load is shared across the worker cores.

Figure \ref{fig:parcegarcores} shows the results of our
evaluation. The topmost points indicated by \texttt{FL} are instances
for which $\parcegar$ timed out.  We can see that, for all $5$
benchmarks, the time taken to synthesize Boolean function vectors when
the number of cores is $2$ is considerable; in fact, $\parcegar$ times
out on three of the benchmarks.  When we increase the number of cores
we observe that (a) by synthesizing in parallel, we can now solve
benchmarks for which we had timed out earlier, and (b) speedups of
about $4-5$ can be obtained with $5-15$ cores. From $15$ cores to $25$
cores, the performance of the algorithm, however, is largely invariant
and any further increase in cores does not result in further speed up.

To understand this, we examined the benchmarks and found that their
AIG representatation have more nodes close to the leaves than to the
root (similar to the DAG in Figure \ref{fig:decomp}). The time taken
to process a leaf or a node close to a leaf is typically much less
than that for a node near the root.  Furthermore, the dependencies
between the nodes close to the root are such that at most one or two
nodes can be processed in parallel leaving most of the cores
unutilized.  When the number of cores is increased from $2$ to $5 -
15$, the leaves and the nodes close to the leaves get processed in
parallel, reducing the overall time taken by the algorithm. However,
the time taken to process the nodes close to the root remains more or
less the same and starts to dominate the total time taken. At this
point, even if the number of cores is further increased, it does not
significantly reduce the total time taken. This behaviour limits the
speed-ups of our algorithm.  For the remaining experiments, the number
of cores used for $\parcegar$ was 20.

\medskip

\noindent\textit{$\bparcegar$ vs $\bcegarsk$:}  
As can be seen from Figure \ref{fig:parsk}, $\cegarsk$ ran
successfully on only $12$ of the $46$ benchmarks, whereas $\parcegar$
was successful on $39$ benchmarks, timing out on $6$ benchmarks and
running out of memory on $1$ benchmark. Of the benchmarks that
$\cegarsk$ was successful on, $9$ belonged to the arithmetic category,
$2$ to the factorization and $1$ to the disjunctive decomposition
category.  On further examination, we found that factorization and
arithmetic benchmarks (except the {\em intermediate} problems) were
conjunctive formulae whereas disjunctive decomposition benchmarks were
arbitrary Boolean formulas.  Since $\cegarsk$ has been specially
designed to handle conjunctive formulas, it is successful on some of
these benchmarks. On the other hand, since disjunctive decomposition
benchmarks are not conjunctive, $\cegarsk$ treats the entire formula
as one factor, and the algorithm reduces to~\cite{Jian,Trivedi}. The
performance hit is therefore not surprising; it has been shown
in~\cite{fmcad2015:skolem} and~\cite{rsynth} that the algorithms
of~\cite{Jian,Trivedi} do not scale to large benchmarks. In fact,
$\cegarsk$ was successful only on the smallest disjunctive
decomposition benchmark.

\begin{figure*}[t]
\centering
\begin{subfigure}{2.4in}
  \centering
  \includegraphics[angle=-90,  scale=0.28] {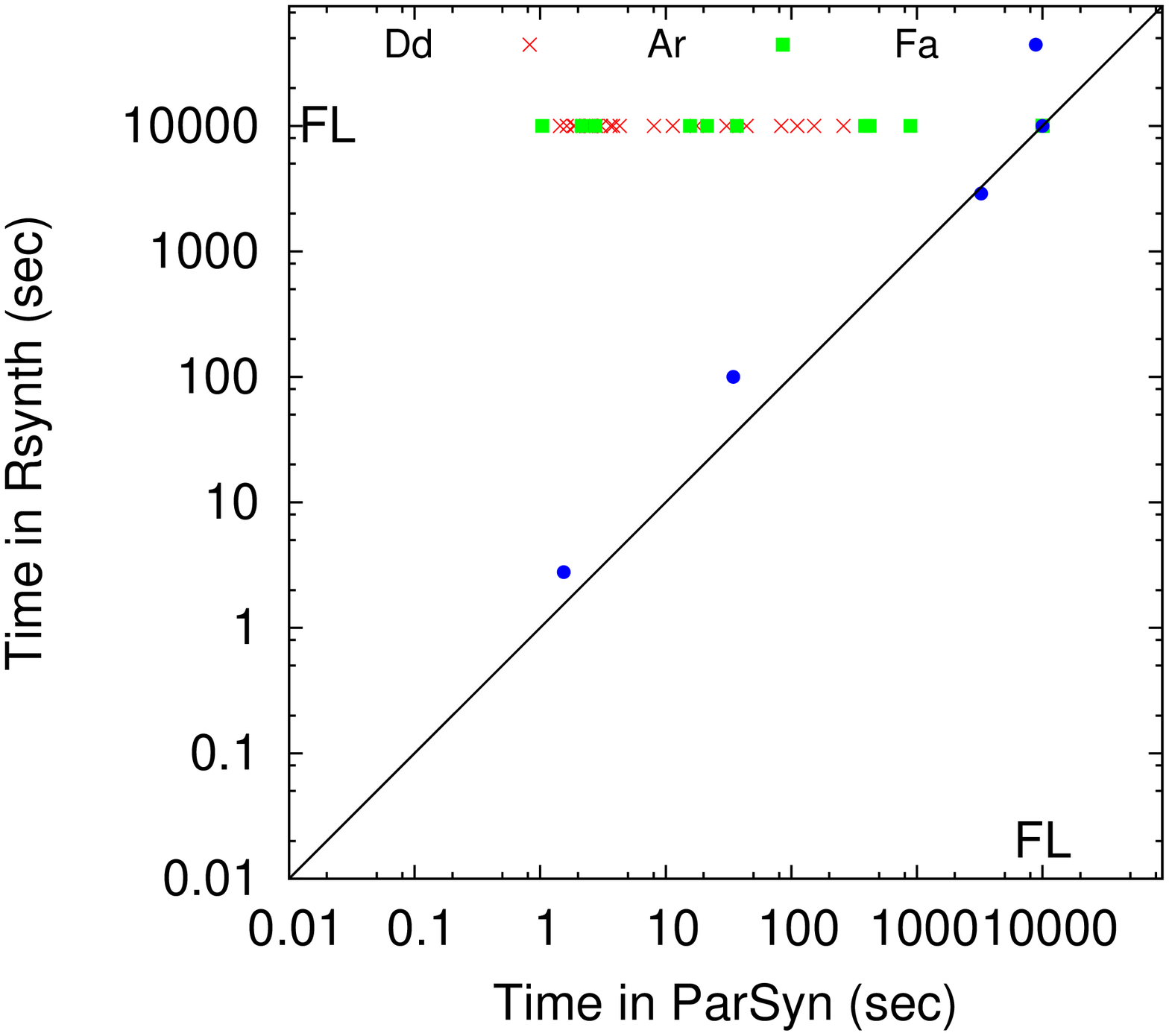}
\caption{$\parcegar$ vs $\rsynth$ } 
\label{fig:cegarrsynth}
\end{subfigure}%
\begin{subfigure}{2.4in}
  \centering
  \includegraphics[angle=-90,  scale=0.28]{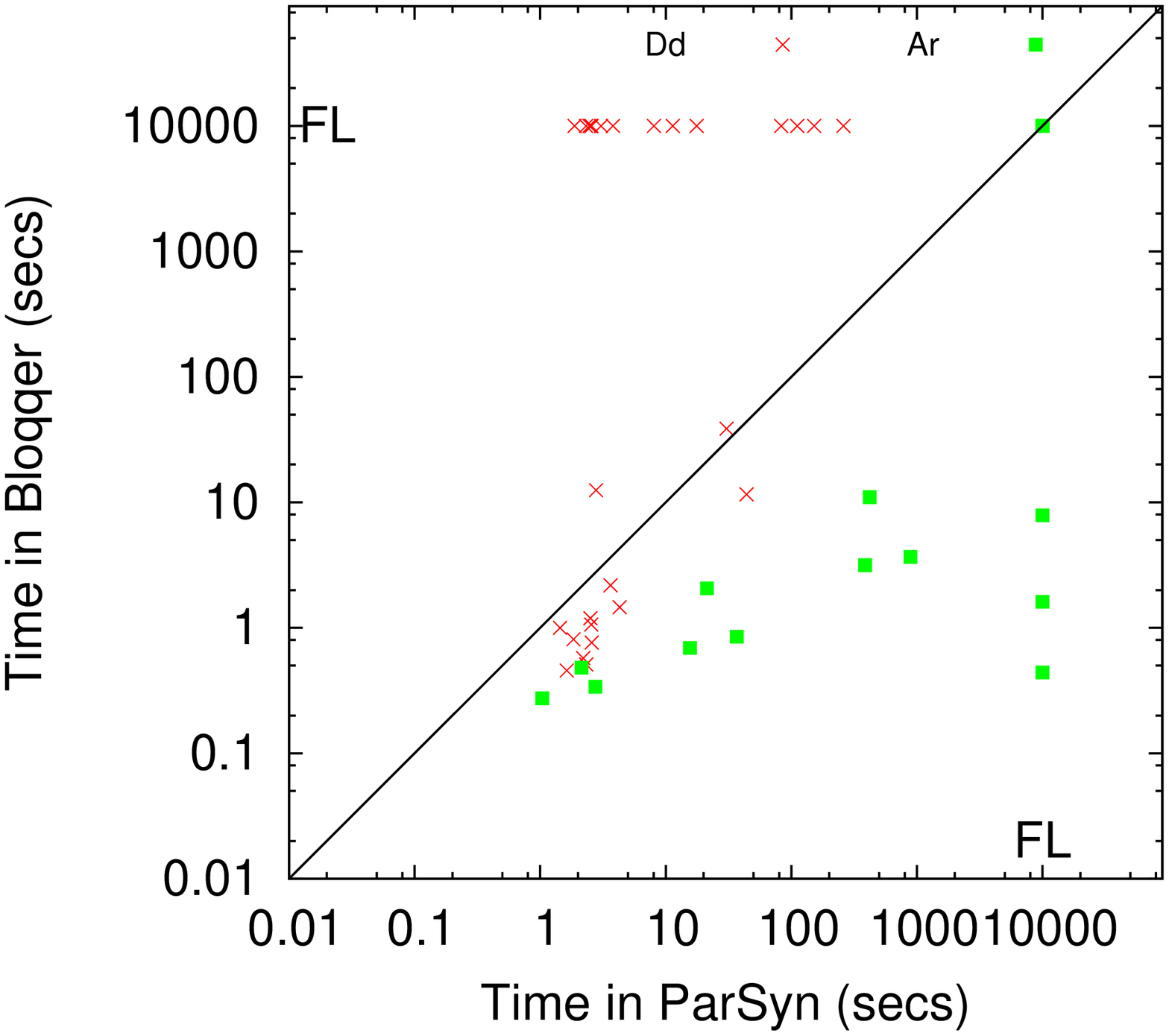}
  \caption{$\parcegar$ vs $\bloqqer$}
\label{fig:cegarbloqqer}
\end{subfigure}
\caption{Legend: \texttt{Ar}: arithmetic, \texttt{Fa}: factorization, \texttt{Dd}: disjunctive decomposition. \texttt{FL}:  benchmarks for which the corresponding algorithm was unsuccessful.}
\end{figure*}
\medskip

\noindent \textit{$\bparcegar$ vs $\rsynth$:}
As seen in Figure \ref{fig:cegarrsynth}, $\rsynth$ was successful only
on $3$ of the $46$ benchmarks; it timed out on $37$ and ran out of
memory on $6$ benchmarks. The $3$ benchmarks that RSynth was
successful on were the smaller factorization benchmarks.  Note that
the arithmetic benchmarks used in \cite{rsynth} are semantically the
same.  In~\cite{rsynth}, custom variable orders were used to construct
the ROBDDs, which resulted in compact ROBDDs.  In our case, we use the
variable ordering heuristic mentioned above (see Sec. 4.1), and
include the considerable time taken to build BDDs from \texttt{cnf}
representation. As mentioned in Section~\ref{sec:introduction}, if we
know a better variable ordering, then the time taken can potentially
reduce.  However, we do not know the optimal variable order for an
arbitrary specification in general.  We also found the memory
footprint of $\rsynth$ to be higher as indicated by the
memory-outs. This is not surprising, as $\rsynth$ uses BDDs to
represent Boolean formula and it is well-known that BDDs can have
large memory requirements.
\medskip

\noindent {\it $\bparcegar$ vs $\bloqqer$:}
Since $\bloqqer$ cannot synthesize Boolean functions for formulas
wherein $\forall X \exists Y \varphi(X,Y)$ is not {\em valid}, we
restricted our comparison to only the disjunctive decomposition and
arithmetic benchmarks, totalling $42$ in number.  From
Figure \ref{fig:cegarbloqqer}, we can see that $\bloqqer$ successfully
synthesizes Boolean functions for $25$ of the $42$ benchmarks.  For
several benchmarks for which it is successful, it outperforms
$\parcegar$.   In line 14 of
Algorithm \ref{alg:composeatop}, \textsc{Perform\_Cegar} makes 
extensive use of the SAT solver, and this is reflected in the time
taken by $\parcegar$. 
However, for the remaining $17$ benchmarks, $\bloqqer$ gave a
{\em Not Verified} message indicating that it could not synthesize Boolean functions for these benchmarks.
In comparison, $\parcegar$ was successful on most of these benchmarks.
\medskip

\noindent{\bf Effect of timeouts on $\bparcegar$.}
Finally, we discuss the effect of the timeout optimization discussed
in Section~\ref{sec:extensions}.  Specifically, for $60$ seconds
(value set through a \emph{timeout} parameter), starting from the
leaves of the AIG representation of a specification, we synthesize
exact Boolean functions for DAG nodes. After timeout, on the
remaining intermediate nodes, we do not invoke the CEGAR step at all,
except at the root node of the AIG.

This optimization enabled us to handle $3$ more benchmarks, i.e.,
$\parcegar$ with this optimization synthesized Boolean function
vectors for all the {\em equalization} benchmarks (in $<$ 340
seconds).  Interestingly, $\parcegar$ without timeouts was unable to
solve these problems. This can be explained by the fact that in these
benchmarks many internal nodes required multiple iterations of the
CEGAR loop to compute exact Boolean functions, which were, however, not
needed to compute the solution at the root node.

\section{Conclusion and future work}
\label{sec:discussion}
In this paper, we have presented the first parallel and compositional
algorithm for complete Boolean functional synthesis from a relational
specification.  A key feature of our approach is that it is agnostic
to the semantic variabilities of the input, and hence applies to a
wide variety of benchmarks.  In addition to the disjunctive
decomposition of graphs and the arithmetic operation benchmarks, we
considered the combinatorially hard problem of factorization and
attempted to generate a functional characterization for it. We found
that our implementation outperforms existing tools in a variety of
benchmarks.

There are many avenues to extend our work. First, the ideas for
compositional synthesis that we develop in this paper could
potentially lead to parallel implementations of other synthesis tools,
such as that described in~\cite{rsynth}.  Next, the 
factorization problem can be generalized to synthesis of inverse functions for
classically hard one-way functions, as long as the function can be
described efficiently by a circuit/AIG.
Finally, we would like to explore improved ways of parallelizing
our algorithm, perhaps exploiting features of specific classes of
problems.


\bibliographystyle{splncs03}

\bibliography{ref}



\end{document}